\documentclass[12pt, draftclsnofoot, onecolumn]{IEEEtran}
\usepackage{mathptmx}

\DeclareSymbolFont{CMlargesymbols}{OMX}{cmex}{m}{n}
\DeclareMathSymbol{\sum}{\mathop}{CMlargesymbols}{"50}
\DeclareMathSymbol{\prod}{\mathop}{CMlargesymbols}{"51}

\DeclareMathAlphabet{\mathcal}{OMS}{cmsy}{m}{n}

\DeclareSymbolFont{Letters}{OML}{cmm}{m}{it}
\DeclareMathSymbol{\alpha}{\mathalpha}{Letters}{11}
\DeclareMathSymbol{\epsilon}{\mathalpha}{Letters}{15}
\DeclareMathSymbol{\lambda}{\mathalpha}{Letters}{21}
\DeclareMathSymbol{\Lambda}{\mathalpha}{Letters}{3}
\DeclareMathSymbol{\pi}{\mathalpha}{Letters}{25}
\DeclareMathSymbol{\rho}{\mathalpha}{Letters}{26}
\DeclareMathSymbol{\sigma}{\mathalpha}{Letters}{27}
\DeclareMathSymbol{\Delta}{\mathalpha}{Letters}{1}
\DeclareMathSymbol{\Phi}{\mathalpha}{Letters}{8}
\DeclareMathSymbol{\Psi}{\mathalpha}{Letters}{9}

\DeclareMathSymbol{Q}{\mathalpha}{Letters}{81}

\DeclareSymbolFont{symbols}{OMS}{cmm}{m}{n}
\DeclareMathSymbol{\infty}{\mathord}{symbols}{"31}

\usepackage{scalerel,stackengine}
\stackMath
\newcommand\wihat[1]{%
\savestack{\tmpbox}{\stretchto{%
		\scaleto{%
			\scalerel*[\widthof{\ensuremath{#1}}]{\kern-.6pt\bigwedge\kern-.6pt}%
			{\rule[-\textheight/2]{1ex}{\textheight}}
		}{\textheight}%
	}{0.5ex}}%
\stackon[1pt]{#1}{\tmpbox}%
}

\usepackage[noadjust]{cite}
\usepackage{yfonts}
\usepackage{empheq}
\usepackage{amsthm}
\usepackage{amsfonts}
\usepackage{amssymb}
\usepackage{amsmath}
\usepackage{tabularx}
\usepackage{comment}
\usepackage{graphicx}
\usepackage{tikz}
\usepackage[caption=false]{subfig}
\usepackage{mathtools}
\usepackage{float}
\usepackage{braket}
\usepackage{hyperref}
\hypersetup
{
	bookmarksdepth=-2,
	linkcolor=black,
	colorlinks=true,
	citecolor=blue
}
\usepackage[mathscr]{eucal}

\usepackage{glossaries}
\glossarystyle{listdotted}
\makeglossaries
\renewcommand*{\CustomAcronymFields}
{
	name={\the\glsshorttok},
	description={\the\glslongtok},
	first={\noexpand\emph{\the\glslongtok}\space(\the\glsshorttok)},%
	firstplural={\noexpand\emph{\the\glslongtok\noexpand\acrpluralsuffix}\space(\the\glsshorttok)},%
	text={\the\glsshorttok},%
	plural={\the\glsshorttok\noexpand\acrpluralsuffix}%
}

\newacronym{rhs}{RHS}{right hand side}
\newacronym{lhs}{LHS}{left hand side}

\newacronym{ap}{AP}{access point}
\newacronym{comp}{CoMP}{Coordinated Multipoint}
\newacronym{sumiso}{SU-MISO}{single-user multiple-input-single-output}
\newacronym{ibl}{IBL}{infinite block length}
\newacronym{fbl}{FBL}{finite block length}
\newacronym{icn}{ICN}{industrial control network}
\newacronym{lan}{LAN}{local area network}
\newacronym{wsn}{WSN}{wireless sensor network}
\newacronym{rt}{RT}{real-time}
\newacronym{tdm}{TDM}{time division multeplxing}
\newacronym{isi}{ISI}{inter-symbol interference}
\newacronym{nist}{NIST}{National Institute of Standards and Technology}
\newacronym{cbrs}{CBRS}{Citizens Broadband Radio Service}
\newacronym{los}{LOS}{line-of-sight}
\newacronym{nlos}{NLOS}{non-line-of-sight}
\newacronym{itu}{ITU}{International Telecommunications Union}
\newacronym{mmwave}{mmWave}{millimeter wave}
\newacronym{nsr}{NSR}{noise-to-signal ratio}
\newacronym{das}{DAS}{distributed antenna system}

\newacronym{csi}{CSI}{channel state information}
\newacronym{cqi}{CQI}{channel quality indicator}
\newacronym{ack}{ACK}{acknowledgement}
\newacronym{arq}{ARQ}{automatic repeat request}
\newacronym{awgn}{AWGN}{additive white Gaussian noise}
\newacronym{cc}{CC}{chase combining}
\newacronym{dp}{DP}{dynamic programming}
\newacronym{fec}{FEC}{forward error correction}
\newacronym{harq}{HARQ}{hybrid automatic repeat request}
\newacronym{hspa}{HSPA}{high speed packet access}
\newacronym{iid}{i.i.d.}{independent and identically distributed}
\newacronym{ir}{IR}{incremental redundancy}
\newacronym{lte}{LTE}{long term evolution}
\newacronym{mdp}{MDP}{markov decision process}
\newacronym{mrc}{MRC}{maximal-ratio combining}
\newacronym{nack}{NAK}{negative acknowledgement}
\newacronym{wimax}{WiMax}{worldwide interoperability for microwave access}
\newacronym{3gpp}{3GPP}{3rd generation partnership project}
\newacronym{ofdm}{OFDM}{orthogonal frequency-division multiplexing}
\newacronym{ofdma}{OFDMA}{orthogonal frequency-division multiple access}
\newacronym{wlan}{WLAN}{wireless local area network}
\newacronym{mmse}{MMSE}{minimum mean-square-error}
\newacronym{gsm}{GSM}{global system for mobile communications}
\newacronym{edge}{EDGE}{enhanced data \gls{gsm} environment}
\newacronym{stbc}{STBC}{space-time block code}
\newacronym{amc}{AMC}{adaptive modulation and coding}
\newacronym{snr}{SNR}{signal-to-noise ratio}
\newacronym{sinr}{SINR}{signal to interference and noise ratio}
\newacronym{mi}{MI}{mutual information}
\newacronym{acmi}{ACMI}{accumulated mutual information}
\newacronym{nacmi}{NACMI}{normalized ACMI}
\newacronym{cdi}{CDI}{channel distribution information}
\newacronym{latr}{LATR}{long-term average transmission rate}
\newacronym{rtr}{RTR}{round transmission rate}
\newacronym{pomdp}{POMDP}{Partially Observable Markov Decision Process}
\newacronym{fd}{FD}{full-duplex}
\newacronym{hd}{HD}{half-duplex}
\newacronym{td}{TD}{Time Division}
\newacronym{tdma}{TDMA}{time division multiple access}
\newacronym{mac}{MAC}{Media Access Control}
\newacronym{uwb}{UWB}{Ultra Wideband}
\newacronym{ieee}{IEEE}{institute of electrical and electronics engineers}
\newacronym{dB}{dB}{decibel}
\newacronym{cdf}{CDF}{cumulative density function}
\newacronym{ccdf}{CCDF}{complementary cumulative density function}
\newacronym{pdf}{PDF}{probability density function}
\newacronym{pmf}{PMF}{probability mass function}
\newacronym{min}{Min.}{minimum}
\newacronym{med}{Med.}{median}
\newacronym{avg}{Avg.}{average}
\newacronym{ul}{UL}{up-link}
\newacronym{dl}{DL}{downlink}
\newacronym{app}{APP}{a-posteriori probability}
\newacronym{logmap}{LogMAP}{log maximum a-posteriori}
\newacronym{llr}{LLR}{log-likelihood ratio}
\newacronym{ue}{UE}{user equipment}
\newacronym{qos}{QoS}{quality of service}
\newacronym{5g}{5G}{5\textsuperscript{th} generation mobile networks}
\newacronym{4g}{4G}{4\textsuperscript{th} generation mobile networks}
\newacronym{tti}{TTI}{transmission time interval}
\newacronym{rrm}{RRM}{radio resource management}
\newacronym{mmib}{MMIB}{mean mutual information per bit}
\newacronym{dsi}{DSI}{decoder state information}
\newacronym{tb}{TB}{transport block}
\newacronym{tbs}{TBS}{transport block size}
\newacronym{cb}{CB}{code block}
\newacronym{cbg}{CBG}{code block group}
\newacronym{cbs}{CBS}{code block size}
\newacronym{prb}{PRB}{physical resource block}
\newacronym{rb}{RB}{resource block}
\newacronym{bler}{BLER}{block error rate}
\newacronym{blep}{BLEP}{block error probability}
\newacronym{crc}{CRC}{cyclic redundancy check}
\newacronym{tdd}{TDD}{time division duplexing}
\newacronym{fdd}{FDD}{frequency division duplex}
\newacronym{embb}{eMBB}{enhanced mobile broadband}
\newacronym{mcc}{MCC}{mission critical communication}
\newacronym{mmc}{MMC}{massive machine communication}
\newacronym{mtc}{MTC}{machine type of communication}
\newacronym{mmtc}{mMTC}{massive machine type of communication}
\newacronym{umtc}{uMTC}{ultra-reliable \gls{mtc}}
\newacronym{urllc}{URLLC}{Ultra-reliable low-latency communications}
\newacronym{rtt}{RTT}{round trip time}
\newacronym{rs}{RS}{reference symbols}
\newacronym{kpi}{KPI}{key performance indicator}
\newacronym{kpis}{KPIs}{key performance indicators}
\newacronym{tx}{Tx}{transmitter node}
\newacronym{rx}{Rx}{receiver node}
\newacronym{cran}{C-RAN}{centralized radio access network}
\newacronym{rru}{RRU}{remote radio unit}
\newacronym{bbu}{BBU}{baseband unit}
\newacronym{fhd}{FHD}{fronthaul delay}
\newacronym{cch}{CCH}{control channel}
\newacronym{saw}{SAW}{stop-and-wait}
\newacronym{qci}{QCI}{\gls{qos} class identifier}
\newacronym{gbr}{GBR}{guaranteed bit rate}
\newacronym{mbr}{MBR}{maximum bit rate}
\newacronym{ngbr}{non-GBR}{non-\gls{gbr}}
\newacronym{arp}{ARP}{allocation and retention priority}
\newacronym{effcr}{ECR}{effective coding rate}
\newacronym{mcs}{MCS}{modulation and coding scheme}
\newacronym{eva}{EVA}{extended vehicular A}
\newacronym{epa}{EPA}{extended pedestrian A}
\newacronym{etu}{ETU}{extended typical urban}
\newacronym{re}{RE}{resource element}
\newacronym{reS}{REs}{resource elements}
\newacronym{nr}{NR}{new radio}
\newacronym{qpsk}{QPSK}{quadrature phase shift keying}
\newacronym{qam}{QAM}{quadrature amplitude modulation}
\newacronym{miso}{MISO}{multiple-input single-output}
\newacronym{mimo}{MIMO}{multiple-input multiple-output}
\newacronym{bs}{BS}{base station}
\newacronym{phy}{PHY}{physical layer}
\newacronym{rlc}{RLC}{radio link control}
\newacronym{bcfsaw}{BCF-SAW}{BCF-SAW}
\newacronym{bcf}{BCF}{backwards composite feedback}
\newacronym{bac}{BAC}{binary asymmetric channel}
\newacronym{bsc}{BSC}{binary symmetric channel}
\newacronym{dtx}{DTX}{discontinued transmission}
\newacronym{bpsk}{BPSK}{binary phase shift keying}
\newacronym{bep}{BEP}{bit error probability}
\newacronym{ndi}{NDI}{new data indicator}
\newacronym{dci}{DCI}{downlink control information}
\newacronym{csit}{CSIT}{channel state information at the transmitter}
\newacronym{lt}{LT}{loudest talker}
\newacronym{ct}{CT}{cooperative transmission}
\newacronym{bpcu}{bpcu}{bits per channel use}

\definecolor{clr_my_dgray}{rgb}{0.4, 0.4, 0.4}
\definecolor{clr_steel_yello}{RGB}{239, 174, 24}
\definecolor{clr_steel_blue}{RGB}{0, 82, 156}
\definecolor{clr_steel_red}{RGB}{199, 3, 45}

\makeatletter
\def\@IEEEsectpunct{:\ \,}
\def\paragraph{\@startsection{paragraph}{4}{\z@}{1.5ex plus 1.5ex minus 0.5ex}%
{0ex}{\normalfont\normalsize\bfseries}}
\makeatother

\newcommand{\given}{\;\middle|\;}

\newcommand*{\medcup}{\mathbin{\scalebox{1.3}{\ensuremath{\cup}}}}
\newcommand{\diag}[1]{\text{diag}\left(#1\right)}

\newcommand{\prob}[1]{\mathbb{P}\left[#1\right]}

\newcommand{\indic}[1]{\mathbf{1}_{\scriptstyle #1}}
\newcommand{\indicof}[1]{\mathbf{1}\,\left\{ #1 \right\}}
\newcommand{\indictwo}[2]{\mathbf{1}_{\scriptstyle #1}\left(#2\right)}

\newcommand{\expect}[1]{\mathbb{E}\left[#1\right]}
\newcommand{\expectw}[2]{\mathbb{E}_{#1}\left[#2\right]}

\newcommand{\lapvarof}[2]{\mathcal{L}\left[#1\right]\left(#2\right)}

\newcommand{\expof}[1]{\exp\left(#1\right)}
\newcommand{\lnof}[1]{\ln\left(#1\right)}
\newcommand{\logof}[1]{\log\left(#1\right)}

\newcommand{\Qmof}[2]{Q_{#2}\left(#1\right)}

\DeclarePairedDelimiter{\norm}{\lVert}{\rVert}

\newcommand{\mb}[1]{\mathbf{#1}}

\newcommand{\tbf}[1]{\widetilde{\mb{#1}}}
\newcommand{\hbf}[1]{\widehat{\mb{#1}}}

\newcommand{\pinf}{{+\infty}}
\newcommand{\ninf}{{-\infty}}
\newcommand{\diff}{\mathrm{d}}

\theoremstyle{definition} 
\newtheorem{proposition}{Proposition}

\newcolumntype{C}{>{\(\displaystyle}c<{\)}@{}} 
\newcolumntype{L}{>{\(\displaystyle}l<{\)}@{}} 
\newcolumntype{R}{>{\(\displaystyle}r<{\)}@{}}

\newcommand{\secref}[1]{Sec.~\ref{#1}}
\newcommand{\figref}[1]{Fig.~\ref{#1}}
\newcommand{\propref}[1]{Prop.~\ref{#1}}
\newcommand{\corref}[1]{Cor.~\ref{#1}}
\newcommand{\tabref}[1]{Table~\ref{#1}}

\newcommand{\plos}{p_{\text{L}}}

\newcommand{\alos}{\alpha_{\text{L}}}
\newcommand{\anlos}{\alpha_{\text{N}}}
\newcommand{\clos}{c_{\text{L}}}
\newcommand{\cnlos}{c_{\text{N}}}

\newcommand{\hhat}{{\hat{h}}}
\newcommand{\htld}{{\tilde{h}}}
\newcommand{\bfh}{{\mathbf{h}}}
\newcommand{\bfhhat}{{\hbf{h}}}
\newcommand{\bfhtld}{{\tbf{h}}}

\newcommand{\evte}{\textsf{TE}}
\newcommand{\evto}{\textsf{TO}}
\newcommand{\evso}{\textsf{SO}}
\newcommand{\evdo}{\textsf{DF}}

\allowdisplaybreaks

\begin{document}

\title{Outage of Periodic Downlink Wireless Networks with Hard Deadlines}
\author{Rebal Jurdi, Saeed R. Khosravirad, Harish Viswanathan, Jeffrey G. Andrews, and Robert W. Heath Jr.
\thanks{Rebal Jurdi (rebal@utexas.edu), Jeffrey G. Andrews (jandrews@ece.utexas.edu), and Robert W. Heath Jr. (rheath@ece.utexas.edu) are with the Wireless Networking and Communications  Group, The University of Texas at Austin.}
\thanks{Saeed R. Khosravirad (saeed.khosravirad@nokia-bell-labs.com) and Harish Viswanathan (harish.viswanathan@nokia-bell-labs.com) are with Nokia - Bell Labs, Murray Hill, NJ.}
\thanks{This work was partially supported by the National Science Foundation under Grant No. NSF-CCF-1514275.}}
\maketitle
\vspace{-1.5cm}
\begin{abstract}
We consider a downlink periodic wireless communications system where multiple \glspl{ap} cooperatively transmit packets to a number of devices, e.g. actuators in an industrial control system.  Each period consists of two phases: an uplink training phase and a downlink data transmission phase. Each actuator must successfully receive its unique packet within a single transmission phase, else an outage is declared.  Such an outage can be caused by two events: a \emph{transmission error} due to transmission at a rate that the channel cannot actually support or \emph{time overflow}, where the downlink data phase is too short given the channel conditions to successfully communicate all the packets. We determine closed-form expressions for the probability of time overflow when there are just two field devices, as well as the probability of transmission error for an arbitrary number of devices. Also, we provide upper and lower bounds on the time overflow probability for an arbitrary number of devices.  We propose a novel variable-rate transmission method that eliminates time overflow.  Detailed system-level simulations are used to identify system design guidelines, such as the optimal amount of uplink training time, as well as for benchmarking the proposed system design versus non-cooperative cellular, cooperative fixed-rate, and cooperative relaying.
\end{abstract}

\begin{IEEEkeywords}
Ultra-reliability, low-latency, mission-critical, machine-to-machine, industrial, IoT, 5G, URLLC, mMTC.
\end{IEEEkeywords}
\section{Introduction}

5G will be the first generation of cellular networks to support mission-critical applications with strict requirements for latency, reliability, and availability \cite{simsek2016,shafi2017}. \gls{urllc}, a work item in 3GPP Release 16, supports such applications in healthcare, automotive, industrial and several other verticals. \Glspl{icn} that perform distributed control, automated scheduling, and predictive maintenance, demand firm reliability and latency requirements. Cables used in these networks are costly, bulky, and obstructive \cite{brooks2001}, which makes wireless connectivity desirable. Unfortunately, existing commercial wireless technologies, such as Bluetooth (IEEE 802.15.1), ZigBee (IEEE 802.15.4), WiFi (IEEE 802.11), and cellular (3GPP Rel-16) fail to provide the needed low latency and high reliability. Industrial systems could either adopt the \gls{urllc} solution, revamp the physical or MAC layers of IEEE 802 protocols to support high reliability and low latency \cite{willig2005}, or implement a clean-slate protocol inspired by legacy wired \glspl{icn}. In this paper, we propose and analyze a transmission method that could be part of a clean-slate wireless \gls{icn}. 


Three characterizing features of \glspl{icn} are periodicity, determinism, and the use of controlled, conflict-free access mechanisms. First, many industrial networks are designed around the periodic transmission of data frames \cite{schickhuber1997}. In general, control traffic has strong periodic patterns \cite{barbosa2012}, as sensor information is retrieved by periodically polling the sensors, and actuators states are periodically updated \cite{thomesse2005}. Second, \glspl{icn} are deterministic, i.e. they are designed to communicate data in a guaranteed time frame \cite{moyne2007}. In real-time control, data must be received by the actuators within a specified time window to meet not only the performance guarantees of control loops, but the also the guarantees of safety functions \cite{moyne2007}. Last, these networks use controlled, conflict-free access mechanisms that are more suitable for periodic traffic than a random-access mechanism. Often, simple token passing or a time slot organization is used. A clean-slate wireless protocol for \glspl{icn} should be tailored around these three defining features.

It is well known that diversity transmission is very effective in achieving high reliability when communicating over fading channels. Typically, time, frequency and spatial degrees of freedom may be exploited to achieve diversity. Apart from time and frequency diversity  which require relying on nature and hence are not reliable sources of diversity gain, 
spatial diversity may be exploited through multiple transmission points or antennas. For low-cost deployments, it is desirable to have a small number of spatially distributed \glspl{ap} with limited number of antennas. Thus it is necessary to achieve spatial diversity through multi-user diversity exploiting the fact that several users are part of the communication system. 
In this paper, we propose to exploit multi-user diversity through feedback. If the transmission spectral efficiency for each user is adjusted according to their channel, then  the total transmission time to send all the packets is dependent on the channels of all the users and benefits from multi-user diversity. In this paper, we explore achieving ultra reliability through this form of multi-user diversity and show that such an approach can be superior to cooperative relaying.

\subsection{Contributions}
The objective of this paper is to model and study the outage (or failure) of an industrial control system that implements a communications mechanism with periodic transmissions, a \gls{tdm} scheduling mechanism, and a strict interpretation of hard-deadline violations as system failures with certainty. Our contributions are summarized as follows.
\begin{itemize}
\item We model a wireless communications method for the transmission of packets of commands from a central controller to a number of actuators through multiple cooperating \glspl{ap}, as might be used in \gls{comp} \cite{randa:2012,randa:2010,amitava:2010} or \glspl{das} \cite{wang:2013,zhang:2008,zhuang:2003}. The system uses \gls{tdd} and operates in two phases: an uplink training phase where the channel state is determined, and a downlink data phase where the commands are dispatched. The transmissions occur periodically assuming a \gls{tdm} scheduling mechanism with variable-length time slots. Additionally, the control commands are to be received by hard deadlines of the cycle boundaries. Failure to deliver these commands by the hard deadlines causes an outage of the industrial system (see \secref{out} for details). We believe that our work is the first to model a periodic, time-orthogonal, multi-user communications method with hard deadlines. 

\item We identify and analyze the events causing a \textit{system outage}, or \textit{outage} for short. The violation of packet delivery by a hard deadline is one source of outage; we refer to this event as \textit{time overflow}. A second cause of outage is the use of transmission rates that is too high to be supported by the wireless channels between the actuators and the \glspl{ap}; we refer to this event as \textit{transmission error}. We analyze the probability of the two events. We find closed-form expressions for the probability of transmission error, a closed-form expression for the probability of time overflow for two users, and upper and lower bounds for the probability of time overflow for an arbitrary number of actuators. We modify the original method to eliminate time overflow and thus reduce the probability of outage.

\item We simulate such a system for a large number of actuators that are arbitrarily scattered  on a factory floor to determine the probability of outage, transmission error and time overflow. We compare the outage probability to the one obtained for benchmark  methods such as cellular, one-shot fixed-rate transmission and two-hop cooperative-relay transmission. Our results show that the variable-rate method outperforms the two benchmarks for a wide range of target throughput.
\end{itemize}

\subsection{Related Work}
We frame our work in the context of analytical approaches to low-latency wireless \glspl{icn} and to \gls{urllc} in general.
Most of the existing analytical approaches to low-latency wireless \glspl{icn} assume a system with spatial or cooperative diversity \cite{willig2008}. 
Different relaying algorithms that use Luby coding were devised in \cite{girs2013} and compared with traditional \gls{arq} and mesh techniques.
A wireless broadcast technique that uses low-rate coding and semi-fixed resource allocation was introduced in \cite{weiner2014}. 
A suite of wireless communications protocols that exploit cooperative diversity through numerous relays and simultaneous retransmissions was introduced in \cite{swamy2015} and later developed and analyzed in subsequent work.
A wireless token-passing protocol that also exploits cooperative diversity was proposed in \cite{dombrowski2015}.
An energy efficient broadcast method for industrial wireless sensor networks was proposed in \cite{chen2017}.
A time-frequency slotted random access uplink method with retransmissions was suggested in \cite{derya17}. 
Finally, the performance of a wireless communications system with imperfect \gls{csi} and finite-length coding was studied in \cite{schiessl2017}. The trade-off between the length of the pilot sequence and that of the data codeword length is characterized in a queuing framework on top of \gls{phy} models, and a rate adaptation strategy is devised. 
Most of this work \cite{willig2008,girs2013,weiner2014,swamy2015,dombrowski2015,chen2017,derya17,schiessl2017} does not account for the difference in channel conditions across different devices, resulting in a conservative choice of transmission rate that caters to the device with the worst channel. Since the actuating devices are scattered around the factory floor, we use variable-length transmission slots to allocate airtime to each device depending on its instantaneous channel state.

Latency and reliability have also been addressed through different analytical frameworks in recent work on \gls{mmwave} massive MIMO networks \cite{vu2017}, \gls{mmwave} dense networks \cite{yang2017}, mobile edge computing \cite{liu2017}, content distribution \cite{huang2017}, and vehicular ad-hoc networks \cite{golnarian2016}. 
In the queuing framework, the average or probabilistic delay is either the object of minimization, or a constraint on optimizing a utility function such as throughput and energy consumption (see \cite{vu2017,huang2017,liu2017}). In the case of a probabilistic delay constraint, the probability of violation of a target delay is constrained by a target reliability parameter.
In the more relevant context of multiple access and broadcast channels (MACs and BCs), a more strict notion of delay constrains the optimization of a utility: all transmitters have a limited number of blocks or slots to deliver their messages in the case of MACs, and all receivers have a limited number of slots to receive their messages in the case of BCs (see \cite{negi2002,tuninetti2001} and \cite{swamy2015} and subsequent work). While prior work \cite{vu2017,yang2017,liu2017,huang2017,golnarian2016,negi2002,tuninetti2001,swamy2015} is on various aspects of low latency and high reliability, the proposed models are unsuited to analyze outages for multiple users at the level of the wireless links and do not harness the periodicity of \glspl{icn} and their controlled medium access protocols.


The rest of the paper is organized as follows. \secref{sys} describes the system model and details the variable-rate communications method. \secref{out} defines the outage events and analyzes their probability. \secref{num} presents numerical results and provides performance insights. Finally, \secref{conc} concludes the paper.

\section{System Model}\label{sys}
\newcommand{\sige}{{\sigma_\text{E}}}
\newcommand{\sigesq}{{\sigma^2_\text{E}}}
\newcommand{\sigefour}{{\sigma^4_\text{E}}}
\newcommand{\varhh}{{\left(1-\sigesq\right)}}
\newcommand{\ctep}{{\frac{2-\sige^2}{\sigesq(1-\sigesq)}}}
\newcommand{\mytd}{{T_\text{D}}}
\newcommand{\mytp}{{T_\text{P}}}
\newcommand{\myts}{{T_\text{S}}}
\newcommand{\mythat}{{\hat{T}_\text{D}}}

In this section, we describe the communications system, and the highlight the main system assumptions we use in our analysis.

\paragraph*{Setup} A controller is wired to $A$ fully-synchronized transmitters/\glspl{ap}. The \glspl{ap} communicate wirelessly with $D\gg A$ field devices scattered on a factory floor. All \glspl{ap} coordinate their transmissions to every device, as the \gls{csi} is shared. These \glspl{ap} behave in a similar way to a cooperating set in \gls{comp} or to transceivers in \gls{das}.

\paragraph*{Resources} Every device expects $B$ bits of data to be delivered every $T$ seconds over a bandwidth of $W$ Hertz. The system uses \gls{tdd} and operates in cycles of period $T$, split into an uplink {training phase of duration $\mytp$ followed by a downlink data phase of duration $\mytd$, $T=\mytp+\mytd$. 

\paragraph*{Channel} Wireless channels linking every \gls{ap}-device pair are assumed to undergo independent Rayleigh fading that is frequency flat, permitting the use of single-carrier modulation without the need of equalization or sequence detection at the receiver. The symbol time is taken to be $\myts=1/W$. While measurement campaigns inside factories indicate that the wireless channel is frequency-selective especially when a wide bandwidth is used \cite{pahlavan89,rappaport91}, we maintain the frequency-flatness assumption for the purpose of analytical tractability. Analysis under frequency-selectivity is deferred to future work.

\paragraph*{Fading dynamics} We assume a quasi-static setting where each packet experiences a single fading value. This setting arises whenever the packets are small enough relative to the fading coherence, in time and frequency, for the fading to remain approximately constant over each packet. The appropriate metric to quantify transmission errors in a quasi-static setting is the outage probability. Additionally, error symbol and packet error probabilities can be identified with the outage probability \cite{lozano2012}, but the latter is more general as it does not depend on the modulation.

\paragraph*{Channel estimation} In the training phase, channel sounding is performed by the \glspl{ap}. The devices sequentially transmit uplink pilot sequences of length $L$, leading to a total duration of $\mytp = L \myts$. For the purpose of analysis, we assume that the devices use the same transmit power $P_\text{T}$ as that used by the \glspl{ap}. Upon receiving the pilot sequences, the \glspl{ap} perform channel estimation using \gls{mmse} for all $D$ channels, one channel at a time. Let $N_0$ be the thermal noise power spectral density, $r$ the transmitter-receiver separation, $\ell(r)$ the signal path loss, and $\rho=P_\text{T}\ell(r)/N_0W$ the average received \gls{snr}. The average received \gls{snr} is averaged with respect to the fading distribution, as the transmit power and path loss are assumed to be known and fixed. For a Rayleigh fading channel, \gls{mmse} channel estimation allows expressing the true channel $h$ as a sum of the channel estimate $\hat{h}$ and the estimation error $\htld$, i.e. $h=\hat{h}+\htld$, where $\htld\sim\mathcal{CN}(0,\sigesq(L,\;\rho))$, $\hat{h}\sim\mathcal{CN}(0,1-\sigesq(L,\;\rho))$, and $\sigesq(L,\;\rho)$ denotes the variance of the channel estimation error \cite{hassibi03,yoo06} given as
\begin{align}\label{eq:mmse}
\sigesq(L,\;\rho)=\frac{1}{1+\rho L}.
\end{align}
To simplify our analysis, we assume that the average \gls{snr} at the receiver is equal to the average \gls{snr} at the transmitter. This is further justified by three facts. First, channel reciprocity holds since \gls{tdd} is used. Second, the \glspl{ap} and devices transmit at the same power. Third, noise at the detection stage is limited to thermal noise that is assumed to have the same power spectral density across all devices and \glspl{ap}. Moreover, we assume that $\rho$ is known by both the \gls{ap} and device.

\paragraph*{Rate selection} Upon channel estimation, the \glspl{ap} share the channel estimates among one another and jointly pick an appropriate transmission rate for every device. Let the scalars $h_d^{(a)}$ and $\hhat_d^{(a)}$ denote the channel between the $a$th \gls{ap} and the $d$th device, and its estimate as computed by the $a$th \gls{ap}. Let the vector $\bfh_d$ denote the channel between the $d$th device and the $A$ \glspl{ap}, i.e. $\bfh_d=[h_d^{(1)} \;\; h_d^{(2)} \;\; \cdots h_d^{(A)}]^{\sf{T}}$, and $\bfhhat_d$ denote its estimate. Let $\rho_d^{(a)}$ be the average \gls{snr} for the pair consisting of the $d$th device and $a$th \gls{ap}, and $\mathbf{G}_d=\diag{\rho_d^{(1)},\rho_d^{(2)},\dots,\rho_d^{(A)}}$. Let $R_d$ be the rate chosen by the \glspl{ap} to transmit to the $d$th device. Let $T_d$ be the airtime given to the $d$th device, i.e. the time required to transmit its data, and $\mythat=\sum_{d=1}^D T_d$ be the aggregate transmission duration. \tabref{t:summary} contains a summary of the notation. The \glspl{ap} take the channel estimate $\bfhhat_d$ at face value, as in conventional adaptive modulation \cite{goldsmith:1997}, and use transmit beamforming to transmit at a rate
\begin{align}\label{eq:rate}
R_d = W\log\left(1+\bfhhat_d^* \mathbf{G}_d \bfhhat_d \right).
\end{align}
The rate $R_d$ is equivalent to the Shannon capacity of an AWGN channel with a \gls{snr} of $\bfhhat_d^* \mathbf{G}_d \bfhhat_d$ and a bandwidth $W$. This particular choice of rate is motivated by the fact that the spectral efficiency envelope of practical \glspl{mcs} has a small spread against the Shannon capacity, which is a smooth function of \gls{snr} that provides an analytical edge.

We make the following remarks:
\begin{itemize}
\item Transmission rates are chosen dynamically, i.e. upon receiving the training sequence and performing channel estimation.
\item Transmission rates vary across devices because every device experiences a unique channel. Additionally, transmission rates are independent random variables due to the independent channel fading assumption.
\item Airtimes vary across devices because they are a function of transmission rates which across devices. Additionally, airtimes are independent random variables.
\item Transmission rates and airtimes vary across cycles due to the quasi-static setting assumption.
\end{itemize}

\paragraph*{Blocklength} Although we assumed that packets are small enough so that fading appears constant, we assume that the packets are large enough so that we can perform our analysis in the \gls{ibl} regime. We defer the analysis in the \gls{fbl} \cite{polyanskiy2010} to future work.

\begin{table*}[t!]
\centering
\caption{Summary of notation.}
\label{t:summary}
\begin{tabularx}{\textwidth}{|c|X|} 
\hline
\textbf{Notation} 			& \textbf{Description} \\
\hline
$A;\;D$						& Total number of transmitting \glspl{ap}; total number of receiving devices.					\\
\hline
$B;\;L$						& Payload size per device, in bits; number of uplink pilots per device. 						\\
\hline
$W;\;\myts$					& Available bandwidth; symbol period. We have $\myts = 1/W$. 									\\
\hline
$T;\;\mytp;\;\mytd$ 		& Cycle duration; training phase duration; downlink phase duration. We have $T=\mytp+\mytd$.	\\
\hline
$h_d^{(a)};\hat{h}_d^{(a)}$ 		& Channel (fading) between the $a$th \gls{ap} and $d$th device; its \gls{mmse} estimate.  		\\
\hline
$\bfh_d;\bfhhat_d;\bfhtld_d$			& Channel vector between the $A$ \glspl{ap} and the $d$th device; its \gls{mmse} estimate; its channel estimation error. 		\\
\hline
$\rho_d^{(a)};\;\rho_d$   		& Average received \gls{snr} of the link between the $a$th \gls{ap} and the $d$th device; average received \gls{snr} of the link between the only \gls{ap} and the $d$th device (when $A=1$). \\
\hline 
$\mathbf{G}_d$				& Diagonal matrix of average \gls{snr} values. We have $\mathbf{G}_d=\diag{\rho_d^{(1)},\rho_d^{(2)},\dots,\rho_d^{(A)}}$. \\
\hline
$C_d;\;R_d;\;T_d$ 			& Mutual information of the channel between the \gls{ap} and the $d$th device; rate chosen by the \glspl{ap} to transmit to the $d$th device; the airtime (transmission duration) given to the 
$d$th device, where $T_d=B/R_d$. \\
\hline
$\mythat$ 					& Total required airtime, i.e. the total duration needed to transmit data to all devices. \\
\hline
$\beta$ 					& Backoff parameter. \\
\hline
$Y_d$						& Scaled airtime $T_d$ used for analytical convenience. It is equivalent to the time needed to transmit one bit in one Hertz. \\
\hline
$G_d;\;\overline{G}_D $ 				& Measurement \gls{nsr} (when $A=1$), $G_d=1/\rho_d(1-\sigesq)$; arithmetic mean of measurement \glspl{nsr}. \\
\hline
$\evdo$						& The event of \textit{device failure}. This occurs when an arbitrary device fails to receive its intended data due to the use of a rate that is unsupported by the channel. \\
\hline
$\evte$ 					& The event of \textit{transmission error}. This occurs when at least one device fails to receive its intended data. \\
\hline
$\evto$ 					& The event of \textit{time overflow}. This occurs when the total required airtime exceeds the downlink time budget. \\
\hline
$\evso$ 					& The event of outage of the industrial system. This is contingent on either $\evte$ or $\evto$. \\
\hline
\end{tabularx}
\end{table*}

\section{Outage Analysis}\label{out}
In this section, we characterize and then analyze the outage probability. The industrial system is said to be in outage when at least one device fails decoding its intended data, or the sum of all transmission durations exceed the allotted downlink time budget $\mytd$.

\paragraph{Transmission error (\evte)} The probability of error experienced by the receiving device is bounded away from zero when the rate chosen by the transmitting \glspl{ap} exceeds the input-output mutual information. In such a scenario, the device decodes its data incorrectly and thus fails to receive the command that is intended by the controller. We say that a transmission error occurs if at least one device decodes its data incorrectly. While this seems restrictive, consider a manufacturing process consisting of numerous workcells (layers of a process) where mechanical devices (e.g. robot arms, pickers, and forklifts) must work harmoniously. A single device failure could jeopardize the operation of the entire process as failures in local workcells can easily cascade to other workcells \cite{nist800}. In the quasi-static setting, the probability of device failure can be written in terms of the outage probability, which is the probability that the transmission rate exceeds the mutual information of the \glspl{ap} and an arbitrary device. Under a per-\gls{ap} power constraint and when transmit beamforming is used, the input-output mutual information is given as in \cite[eq. 11 ]{wang:2013}
\begin{align}
C = W\log\left(1+\bfh^* \bf{G} \, \bfh \right),
\end{align}
where $\mathbf{G}=\diag{\rho^{(1)},\rho^{(1)},\dots,\rho^{(A)}}$, and the probability of device failure corresponding to a selected rate $R$ based on the channel estimate is
\begin{align}\label{eq:evdo}
\prob{\evdo} = \prob{R>C}.
\end{align}
The probability of transmission error is expressed in terms of the probability of device failure as
\begin{align}\label{eq:df_to_te}
\prob{\evte} 
= \prob{\bigcup_{d=1}^D\left\{R_d>C_d\right\}} 
= 1-\prod_{d=1}^D \left(1-\prob{R_d>C_d}\right),
\end{align}
due to independent channel fades.

\paragraph{Time overflow (\evto)}
In many \gls{rt} applications, actuators must change their states isosynchronously to ensure a smooth process \cite{neumann2007}. Translating this requirement into our model, all data must be decoded and all instructions applied simultaneously at cycle edges. Otherwise, the industrial system is prone to outage. The communication method that we have modeled does not guarantee that all devices' instructions are delivered within the same cycle. If enough actuators are in deep fade or if path loss is too high, the total airtime might overflow into a subsequent cycle, i.e. the aggregate transmission duration $\mythat$ might exceed the downlink budget $\mytd$. We accordingly define the probability of time overflow as
\begin{align}
\prob{\evto} = \prob{\mythat > \mytd} = \prob{\sum_{d=1}^D \frac{B}{R_d} > \mytd}.
\end{align}

\paragraph{System Outage (\evso)}
Combining the above two events, we can  express the outage probability as
\begin{align}\label{eq:evso}
\prob{\evso} 
= \prob{\evte \; \medcup \; \evto} 
= \prob{\bigcup_{d=1}^D \left\{R_d > C_d\right\}\;\medcup\;\left\{\sum_{d=1}^D \frac{B}{R_d} > \mytd\right\}}.
\end{align}
The outage probability is difficult to calculate because the events $\evte$ and $\evto$ are not independent. We analyze $\prob{\evte}$ and $\prob{\evto}$ separately and compute $\prob{\evso}$ using Monte Carlo simulation. 

\subsection{Transmission Error}\label{ss:te}
Calculating the probability of device failure $\prob{\evdo}$ is a prerequisite for calculating the probability of transmission error $\prob{\evte}$ (see \eqref{eq:df_to_te}). Calculating $\prob{\evdo}$ requires averaging the function $\indic{R>C}$ over the joint density of $\bfhtld$ and $\bfhhat$. Using the law of total expectation, and exploiting the fact that the channel estimate $\bfhhat$ and the estimation error $\bfhtld$ are independent, we first calculate the conditional probability of device failure
\begin{equation}
\prob{\evdo\given \bfhhat} =  \expectw{\bfhtld}{\indic{R>C}\given \bfhhat}
\end{equation}
then the probability of device failure is
\begin{equation}
\prob{\evdo} = \expectw{\bfhhat}{\prob{\evdo\given \bfhhat}}.
\end{equation}

The conditional probability $\prob{R>C\given \bfhhat}$ has been widely studied in the context of MISO and SIMO communication over Rayleigh channels with imperfect \gls{csi} \cite{xie:2005,shin:2008,zhang:2008}. For an arbitrary rate $R$, $\prob{R>C\given \bfhhat}$ reduces to evaluating the \gls{cdf} of a non-central Chi-squared distribution with $2A$ degrees of freedom, which is expressed as a Marcum $Q$-function of order $A$ as
\begin{equation}
Q_A(u,v) = \frac{1}{u^{A-1}}\int\limits_v^\pinf x^A \expof{-\frac{x^2+u^2}{2}}\;I_{A-1}(ux)\; \diff x,
\end{equation}
where $I_\nu(\cdot)$ is the modified Bessel function of the first kind and order $\nu$. The Marcum $Q$-function also appears in the error probability analysis of Rician channels, multichannel communications, and radar communications \cite{marcum:1960}. For $M=1$, we have
\begin{equation}
C=\logof{1+\rho|h|^2};\;\;\;R=\logof{1+\rho|\hhat|^2}.
\end{equation}
As a result, $\prob{R>C\given \bfhhat}$ is given similarly to \cite[eq.~(9)]{zhang:2008} as
\begin{equation}
1-\Qmof{\dfrac{\sqrt{2}}{\sige}|\hhat|,\dfrac{\sqrt{2}}{\sige}|\hhat|}{1}.
\end{equation}
Using the identity \cite{marcum}
\begin{equation}
Q_1(u,u) =\frac{1}{2}\left(1+e^{-u^2}I_0(u^2)\right),
\end{equation}
the conditional probability simplifies to
\begin{equation}
\prob{R>C\given \hhat} 
= \frac{1}{2}\left[1-e^{-\frac{2}{\sigesq}|\hhat|^2}\,I_0\left(\frac{2}{\sigesq}|\hhat|^2\right)\right].
\end{equation}
Now, we remove the conditioning on the channel estimate to determine the probability of device failure on average. 
\begin{proposition}[Device failure, $A=1$]\label{p:err_1r_1t}
For $A=1$. The probability of device failure is
\begin{equation}
\prob{R>C} 
=
\frac{1}{2} \left(1 - \frac{\sige}{\sqrt{4-3\sigesq}} \right).
\end{equation}
\end{proposition}
\begin{proof}
See Appendix~\ref{a:a}.	
\end{proof}

For $A>1$, the conditional probability of device failure is
\begin{equation*}
\prob{\sum_{a=1}^A \rho^{(a)} |h^{(a)}|^2 < 2^R-1 \given \hhat^{(1)},\dots,\hhat^{(A)}}.
\end{equation*}
The random variable $\sum_{a=1}^A \rho^{(a)} |h^{(a)}|^2=\bfh^* \bf{G} \, \bfh$ follows a quadratic form of Gaussian random variables which has a generalized non-central Chi-squared distribution. The expression of its \gls{cdf} is complicated, and often times the inversion formula is used to numerically invert its characteristic function \cite{davies:1980}. Because of the intractability of this distribution, we restrict our analysis to the case where $\rho^{(a)}=\rho$ and defer the results of the more general case to \secref{num}. The following proposition gives the probability of device failure for $A>1$.
\begin{proposition}[Device failure, $A>1$]\label{p:err_1r_Dt}
Consider now $A>1$ and $\rho^{(a)}=\rho$. The probability of device failure is
\newcommand{\rebaltw}{{\frac{\sigesq}{4-3\sigesq}}}
\newcommand{\rebalth}{{\frac{2-\sigesq}{\sqrt{4\sigefour-3\sigesq}}}}
\begin{align}
\begin{split}
\prob{R>C}
&= \frac{1}{2} - \frac{1}{2} \left[\rebaltw\right]^{\frac{A}{2}} \, \left|P_{A-1}\left(\rebalth\right)\right| \\
&- \sum_{a=1}^{A-1}  \frac{(A+a-1)!}{(A-1)!} \left[\rebaltw\right]^{\frac{A}{2}} \, \left|P_{A-1}^{-a}\left(\rebalth\right)\right|, 
\end{split}
\end{align}
where $P_\mu^{-\nu}$ is the associated Legendre function. 
\end{proposition}
\begin{proof}
See Appendix~\ref{a:a}.	
\end{proof}
\noindent The probability of device failure in the case of multiple \glspl{ap} is an algebraic function that can be expressed in terms of a finite number of addition, multiplication, and square-root operations. For example, for $A=2$, 
\begin{align}
\prob{R>C}
= \frac{1}{2} - \frac{6\sigesq-5\sigefour}{\left(8-6\sige\right)\sqrt{4\sige-3\sigesq}}.
\end{align}
\figref{fig:do_vs_sigma} shows the probability of device failure as a function of the estimation error variance ($\sigesq$). We observe that the more \glspl{ap}, the faster the decay of the probability of device failure. Of particular interest, the higher the accuracy of estimation, the higher the chances of device failure. Since $\sigesq$ and $L$, the number of pilots per device, are inversely related, observation suggests that more training increases the chances of device failure. Consequently, the probability of transmission error converges to 1 as $N$ grows large. The reasoning is that channel estimate $\bfhhat$ becomes more accurate with increasing the training length. Moreover, the symmetric nature of the Normal distribution makes the events $\norm{\bfhhat}\leq\norm{\bfh}$ and $\norm{\bfhhat}>\norm{\bfh}$ equally likely with increased training. To circumvent this, i.e. to ensure that longer training leads to a lower transmission error probability, we introduce a multiplicative backoff parameter, $0\leq\beta\leq1$, from the rate in \eqref{eq:rate} so that it becomes
\begin{align}\label{eq:rate_bkoff}
R_d = \beta  W\log\left(1+\bfhhat_d^* \mathbf{G}_d \bfhhat_d \right).
\end{align}
A multiplicative, pre-$\log$ parameter has the following operational significance: for a bandwidth $W$, $\beta W$ is used to transmit information/data bits, and $(1-\beta)W$ is used to transmit control or error correction bits. Alternatives to the pre-log factor have been proposed in prior work. A \textit{utilization factor} that scales the \gls{snr} (inside the $\log$-function) was used for \textit{goodput} maximization in \cite{akoum:2011,zhang:2008,cao:2009,vakili:2006} as a parameter reflecting faithfulness in the quality of estimation. A \textit{scheduling backoff} parameter was used in the same way in \cite{kang:2014}. While a backoff parameter inside the $\log$-function indicates using a fraction of the available power, a backoff parameter outside the $\log$-function indicates using a fraction of the available time or bandwidth. In \secref{num}, we see that backoff optimization improves the outage probability for a practical range of payload size and training sequence length.

\begin{figure*}
\centering
\includegraphics[width=0.4\textwidth]{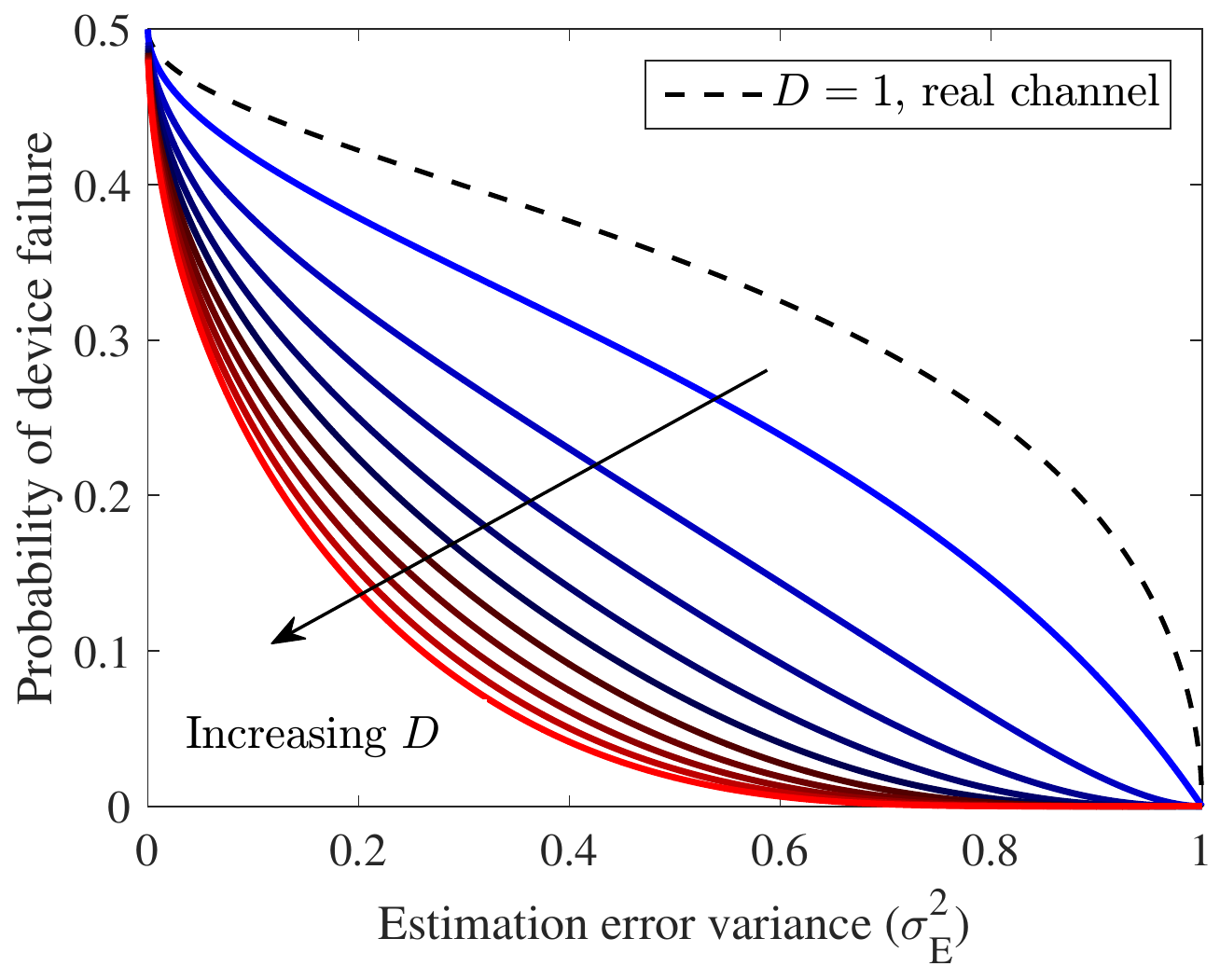}
\caption{Probability of device failure vs. estimation error variance ($\sigesq$) for $D=1,\dots,10$ and for the real Rayleigh channel with $D=1$ (i.e. when the channel takes real values).}
\label{fig:do_vs_sigma}
\vspace{-20pt}
\end{figure*}

\subsection{Time Overflow}\label{ss:to}
Characterizing the probability of time overflow is more challenging than characterizing. As a result, we provide the general expression and bounds in this section. We first introduce the random variable $Y_d=W/R_d$ as the time needed to transmit one bit in one Hertz. Note that $Y_d$ is proportional to the airtime given to the $d$th device which is $T_d=B/R_d$. We recast the probability of time overflow as
\begin{align}
\prob{\evto} = \prob{\sum_{d=1}^D Y_d > \frac{\mytd W}{B}}.
\end{align} 
Unlike the probability of transmission error, the probability of time overflow cannot be expressed in terms of the \glspl{cdf} of the random variables $\set{Y_d}$ as in \eqref{eq:df_to_te}, but rather through the convolution of their \glspl{pdf} as
\begin{align}
\begin{split}
\prob{\evto} 
&= \int\limits_{\mathbb{R}^D_+} \indictwo{\Delta^c}{\mb{y}}\; f_{\mathbf{Y}}(\mb{y}) \; \diff \mb{y} 
= \idotsint\limits_{\mathbb{R}^D_+ \setminus \Delta} f_{Y_1}(y_1)\cdots f_{Y_D}(y_D) \; \diff y_1 \dots \diff y_D,
\end{split}
\end{align}
where $f_{Y_n}$ is the \gls{pdf} of $Y_n$, $f_{\mathbf{Y}}$ the joint \gls{pdf}, and $\Delta$ the tetrahedron
\begin{equation*}
\Delta = \left\{(y_1,\dots,y_D)\in\mathbb{R}^D_+ \;:\; \sum_{d=1}^D y_d \leq \frac{\mytd W}{B} \right\}.
\end{equation*}
Since $\set{Y_d}$ have an infinite first moment, Markov-type inequalities, concentration inequalities, and the central limit theorem cannot be invoked to obtain bounds on the probability of time overflow. While characterizing $\prob{\evto}$ is difficult, the conditional distribution $\prob{\evto\given\set{\bfhhat_d}}$ can be computed as
\begin{align}
\prob{\evto\given\set{\bfhhat_d}_{d=1}^D} = \indicof{\sum_{d=1}^D \frac{1}{R_d} > \frac{\mytd}{B}\given\set{\bfhhat_d}_{d=1}^D},
\end{align}
where $\indicof{\cdot}$ is the indicator, since $R_d$ is a function of $\bfhhat_d$. In other words, one can simply obtain the values of $\set{R_d}$ from $\set{\bfhhat_d}$ and check if $\sum_d 1/R_d > \mytd/B$.

It is more convenient to determine $\prob{\evto}$ through $\prob{\sum_{d=1}^D Y_d \leq y}$, which gives the probability of \emph{time underflow} with respect to a time budget $y$. Note that the budget we are interested in is $y=\mytd W/B$. We first approach $\prob{\evto}$ for the simplest case: $D=2$. The next result gives the exact probability $\prob{Y_1+Y_2\leq y}$ for an arbitrary $y$ as well as a lower bound for this probability.  
\begin{proposition}[Time underflow, two devices]\label{p:to_2r}
Suppose that there is a single \gls{ap} transmitting to two devices that are at a nominal \gls{snr} $\rho$ and consider a deadline $y$. The probability of time underflow is
\begin{equation}\label{eq:to_2r}
\prob{Y_1+Y_2\leq y} 
= \frac{1}{\rho\varhh} \int\limits_{2^{1/y}}^\pinf
\expof{ -\frac{v+v^{1/y\log v - 1}-2}{\rho\varhh} } \, \diff v,
\end{equation}	
which can bounded from above and from below as
\begin{equation}\label{eq:to_2r_bnd}
\expof{-2\,\frac{2^{2/y}-1}{\rho\varhh}}
\leq  
\prob{Y_1+Y_2\leq y} 
\leq 
\expof{-2\,\frac{2^{1/y}-1}{\rho\varhh}}.
\end{equation}
\end{proposition}
\begin{proof}
See Appendix~\ref{a:b}.	
\end{proof}

For a general $D$, the exact expression of time underflow involves an integration over $D$ dimensions, which is neither analytically nor computationally tractable as $D$ grows large. The bounds derived for $D=2$ are useful to derive upper and lower bounds for a general $D$ which are given in the next proposition. Before proceeding, we denote by $G$ the measurement \gls{nsr}, i.e. the inverse of the measurement \gls{snr} $\rho(1-\sigesq)$.
\begin{proposition}[Time underflow, $D$ devices]\label{p:to_nr}
Suppose that there is a single \gls{ap} transmitting to $D$ devices  at a nominal \gls{snr} $\rho$ and consider a deadline $y$. The probability of time underflow is bounded as
\begin{equation}\label{eq:to_nr_bnd}
\expof{-DG\left(2^{D/y}-1\right)}
\leq 
\prob{\sum_{d=1}^D Y_d \leq y}
\leq  
\expof{-DG\left(2^{1/y}-1\right)},
\end{equation}
\end{proposition}
\begin{proof}
The proof uses mathematical induction. The expression is already shown to hold for the base case ($D=2$) in \propref{p:to_2r}. In Appendix~\ref{a:b}, we prove that it holds for arbitrary $D$.	
\end{proof}	

We also give a similar result for when the devices do not necessarily experience the same \gls{snr}. First define $\rho_d$ and $\sigma_d^2$ to be the SNR and estimation error variance for the $d$th device. Also define $G_d$ to be the inverse of $\rho_d(1-\sigma_d^2)$.
\begin{proposition}[Time underflow, $D$ devices at different \gls{snr}]\label{p:to_nr_varsnr}
Out of $D$ devices, suppose that the $k$th device has a measurement \gls{nsr} $G$. The probability of time underflow is upper bounded as
\begin{align}\label{eq:sandwich_D}
\expof{- D\,\overline{G}_D (2^{D/y}-1)  }
&\leq 
\prob{\sum_{d=1}^d Y_d \leq y} 
\leq  
\expof{- D\,\overline{G}_D (2^{1/y}-1)  },\;\text{where} \\
\overline{G}_D 
&= \frac{1}{D}\sum_{d=1}^d G_d.
\end{align}
\end{proposition}
\begin{proof}
The proof follows that of \propref{p:to_2r} and \propref{p:to_nr}. First, prove the result for two devices as in \propref{p:to_2r}, then use recursion to generalize to an arbitrary $D$ as in \propref{p:to_nr}.
\end{proof}

\begin{proposition}[Tight bound on the probability of time underflow, two devices]\label{p:to_2r_bnd}
Suppose that there is a single \gls{ap} transmitting to two devices that are at a nominal \gls{snr} $\rho$ and consider a deadline $y$. The probability of time underflow can be tightly bounded above as
\begin{align}\label{eq:to_2r_tight}
\prob{Y_1+Y_2\leq y}  
\leq \left(\frac{J}{J+G}\right) e^{-2G(2^{1/y}-1)},
\text{where}\\
J 
= \max_{v\geq 2^{1/y}} \left(\frac{1}{2^{1/y}-v}\right)\lnof{1- \expof{G 2^\frac{1}{y}-G 2^\frac{\ln v}{y\ln v - \ln 2}}}.
\label{eq:to_2r_tight_2}
\end{align}	
\end{proposition}
\begin{proof}
See Appendix~\ref{a:c}.	
\end{proof}
The tightened bound is almost identical to the looser bound in \eqref{eq:to_2r_bnd}, up to a shrinking scale factor in the expression \eqref{eq:to_2r_tight}. Note that the value of $J$ can be obtained either by solving for a zero derivative of function on the right hand side of \eqref{eq:to_2r_tight_2}, or by using a simple gradient ascent method. This bound can be generalized to an arbitrary $D$ and to an arbitrary (but known) set of measurement \glspl{nsr} $\set{G_d}$.

\begin{proposition}[Tight bound on the probability of time underflow, $D$ devices]\label{p:to_nr_varsnr_tight}
Out of $D$ devices, suppose that the $d$th device has a measurement \gls{nsr} $G_d$. The probability of time underflow is bounded as
\begin{align}\label{eq:to_nr_bnd_tight}
\prob{\sum_{d=1}^D Y_d \leq y} 
&\leq \expof{-D\overline{G}_D \left(2^{1/y}-1\right)} \cdot \prod_{d=2}^D \frac{J_d}{J_d+G_d},\;\text{where}\\
J_d 
&= \max_{v\geq 2^{1/y}} \left(\frac{1}{2^{1/y}-v}\right)
\lnof{1-\expof{(d-1)\overline{G}_{d-1}\left(2^\frac{1}{y}-2^\frac{\ln v}{y\ln v - \ln 2}\right)}}.
\end{align}
\end{proposition}
\begin{proof}
Similar to the proofs of \propref{p:to_nr_varsnr} and \propref{p:to_2r_bnd}.
\end{proof}

\begin{figure}[t]
\centering
\subfloat[Base SNR $=15$ dB]{
\includegraphics[width=0.3\textwidth]{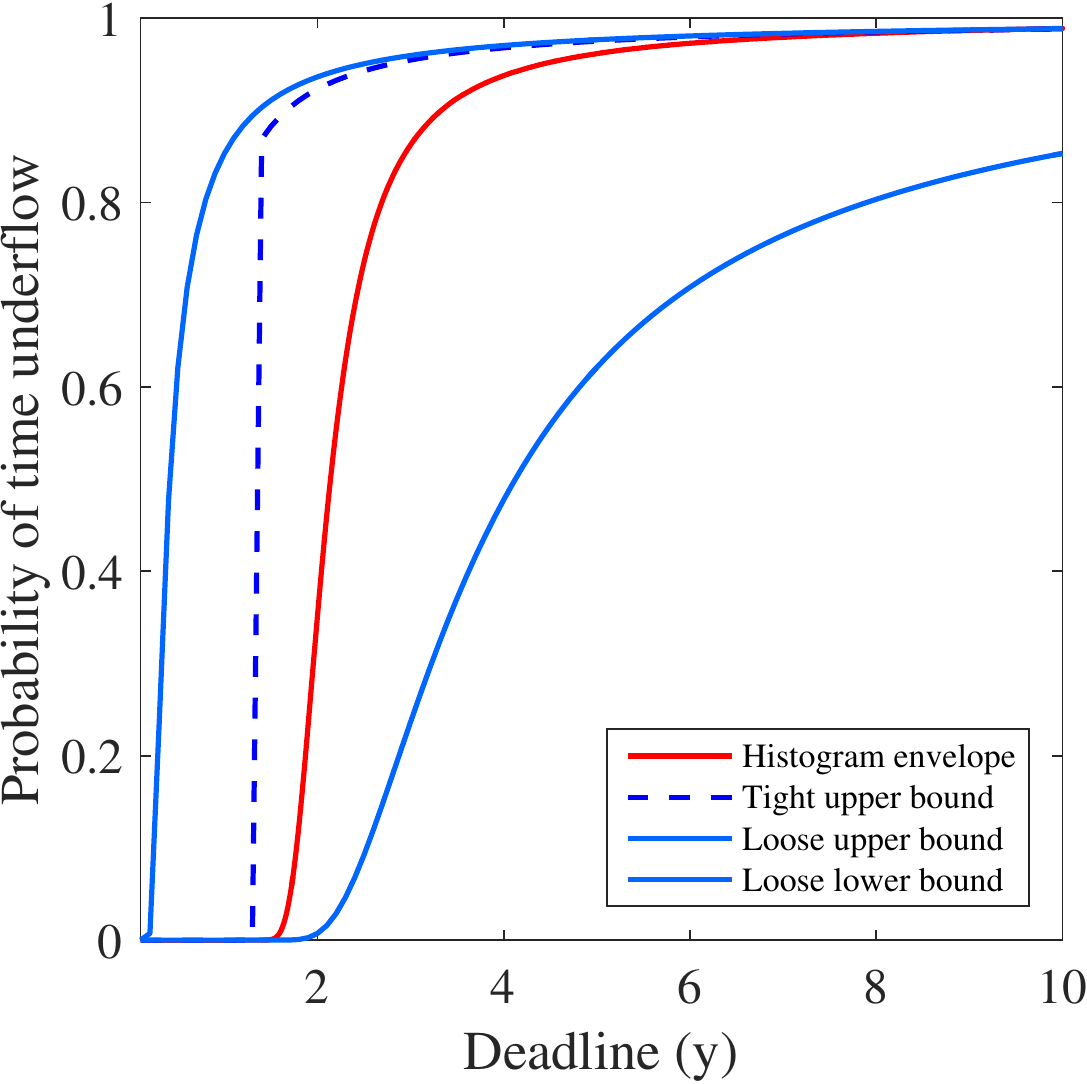}\label{fig:bnd15}
}
\hfill
\subfloat[Base SNR $=20$ dB]{
\includegraphics[width=0.3\textwidth]{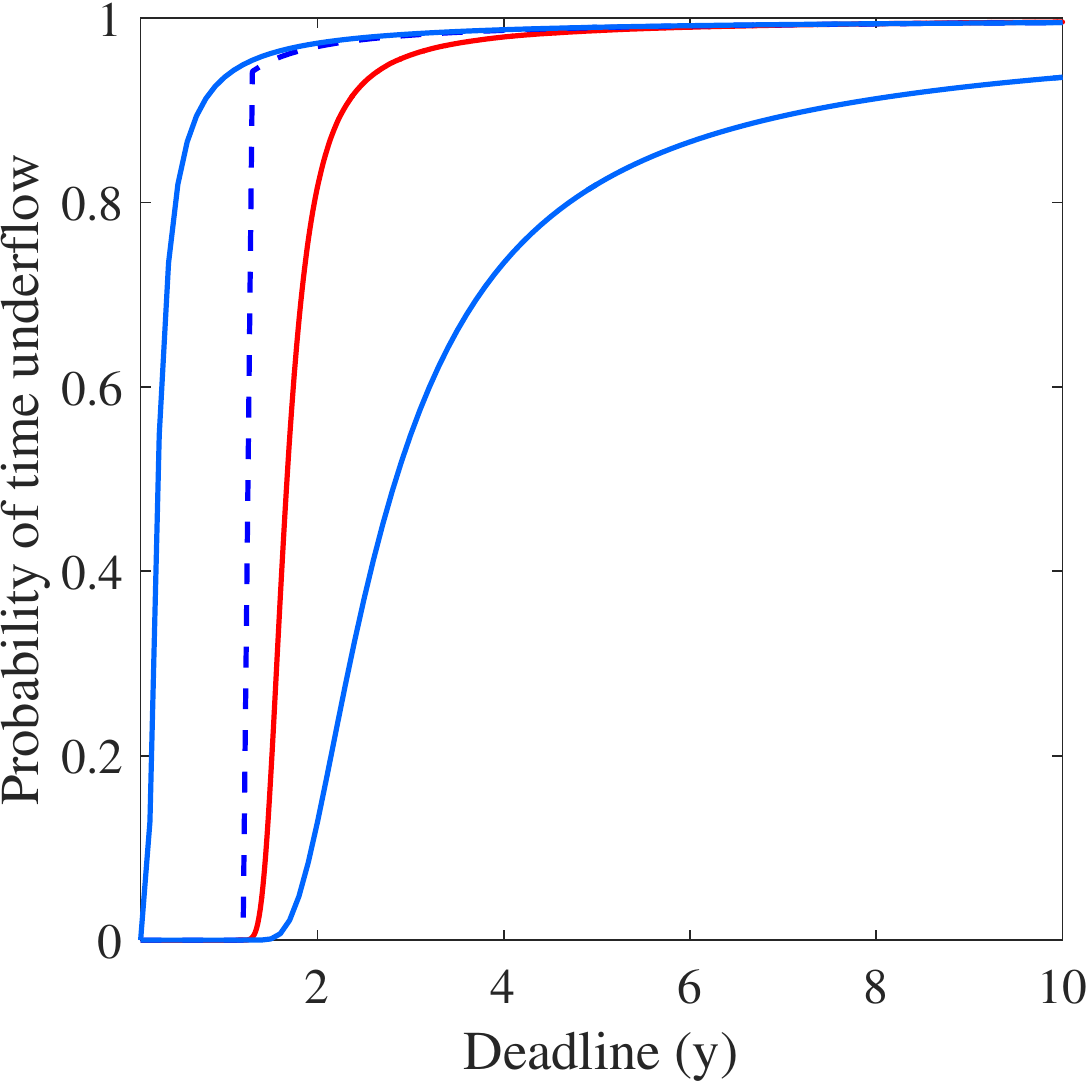}\label{fig:bnd20}
}
\hfill
\subfloat[Base SNR $=25$ dB]{
\includegraphics[width=0.3\textwidth]{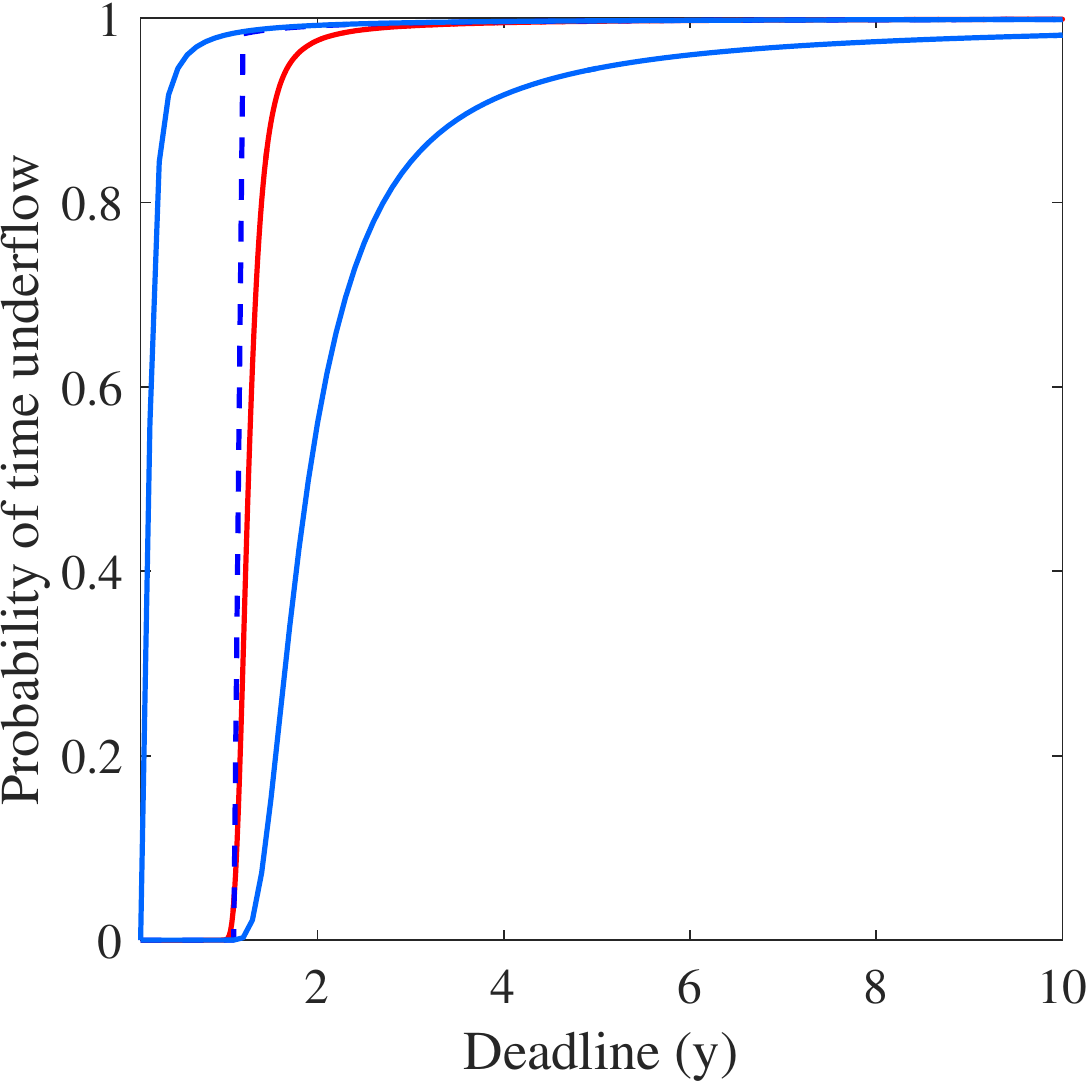}\label{fig:bnd25}
}
\caption{The probability of time underflow as compared to its loose upper and lower bounds in \eqref{eq:to_nr_bnd} and its lower bound in \eqref{eq:to_nr_bnd_tight}. The spread of the tight upper against the histogram envelope narrows down as the base SNR increases.}\label{fig:compare_bounds}
\vspace{-20pt}
\end{figure}

\figref{fig:compare_bounds} compares The probability of time underflow as compared to its loose upper and lower bounds in \eqref{eq:to_nr_bnd} and its lower bound in \eqref{eq:to_nr_bnd_tight}. We consider 1 \gls{ap} and 10 actuators with \glspl{snr} that are uniformly chosen in the range $[\text{Base SNR},\,\text{Base SNR}+5]$ dB. The probability of time overflow is generated experimentally via Monte Carlo simulation, and its plot is the envelope of the empirical probability histogram. We observe that he spread of the tight upper against the histogram envelope tightens as the base SNR increases, suggesting that a higher base SNR improves the approximation of the time underflow probability by its tight upper bound.

All of the previous results have expressed the outage probability when there is a single \gls{ap}, i.e. $A=1$. We also derive in the next proposition a lower and and an upper bound that sandwich $\prob{\evso}$. These bounds hold even when $A>1$.
\begin{proposition}[Sandwich bound for the probability of time overflow]\label{p:sandwich}
For an arbitrary number of \glspl{ap} $A$ and arbitrary (not necessarily equal, but known) \gls{snr} values, the probability of time overflow can be sandwiched as
\begin{align}
\prod_{d=1}^D \prob{R_d \leq \frac{DB}{\mytd}}
\leq
\prob{\evto}
\leq
1-\prod_{d=1}^D \prob{R_d > \frac{DB}{\mytd}}.
\end{align}
\end{proposition}
\begin{proof}
See Appendix~\ref{a:d}.
\end{proof}

\subsection{Reducing the outage probability}\label{ss:reduce}
\newcommand{\mytbreve}{{\breve{T}_\text{D}}}
We propose a variation of the variable-rate method that eliminates time overflow and decreases the probability of outage. We refer to this variation as \textit{modified variable-rate}. Instead of solving for the optimal $\beta$ in \eqref{eq:rate_bkoff} and selecting a rate $R_d$ as in \eqref{eq:rate}, the \glspl{ap} choose $\beta = 1$  and select a rate $\breve{R}_d = \alpha R_d$, where alpha is 
\begin{align}\label{Eq:alpha_modified}
\alpha = \frac{\sum_{d=1}^D T_d}{\mytd}.
\end{align}
The new choice of rate guarantees that the new total airtime $\mytbreve$ will not exceed the allotted downlink budget $\mytd$. Scaling  $R_d$ by $\alpha$ as in \eqref{Eq:alpha_modified} eliminates the need to optimize $\beta$ by noting that the new rate $\breve{R}_d$ is not a function of $\beta$, as $R_d$ is proportional to both $\beta$, as it appears in $R_d$, and $1/\beta$, as it appears in $T_d$.
The probability of outage will decrease in the following two scenarios.
\begin{itemize}
\item  Underspent downlink budget: Suppose that the total airtime $\mythat$ is shorter than the downlink budget $\mytd$, then $\alpha \leq 1$, and $\breve{R}_d \leq R_d$. Therefore, the new probability of device failure $\prob{\evdo}=\prob{\breve{R}_d>C_d}$ could not be worse than what it originally was, $\prob{R_d>C_d}$. Consequently, the probability of transmission error diminishes, and the probability of outage diminishes. 

\item Overspent downlink budget: Suppose that the total airtime $\mythat$ is longer than the downlink budget $\mytd$, then $\alpha > 1$, and $\breve{R}_d > R_d$. In this scenario, there is a time overflow event and thus an outage event regardless of whether there is a transmission error event. Selecting a set of higher rates $\set{\breve{R}_d}$ has one of two possible outcomes: $\breve{R}_d \leq C_d$ for every $d$, or there is a $d$ such that $\breve{R}_d > C_d$. In the former outcome, there is no transmission error and thus no outage. In the latter outcome, there is a transmission error and an outage. But there already was an outage when the set of old rates $\set{R_d}$ was used. We conclude that using the new set of rates $\set{\breve{R}_d}$ can only decrease the outage probability.
\end{itemize}


The analytical results obtained in this section are concerned with special cases: one transmitting \gls{ap} or one \gls{snr} for all devices. Moreover, our analysis provides closed-form expressions for the bounds on the probability of time overflow rather than the exact probability expression. Therefore, we turn to Monte-Carlo simulation to study the probability of outage events for a more general setting: randomly distributed field devices and multiple \glspl{ap}. Additionally, we compare the variable-rate outage probability to that of benchmark methods, and we compare the variable-rate outage probability to that of the modified variable-rate method. The results of this simulation are given in the next section along with a corresponding discussion.

\section{Numerical Results}\label{num}
In this section, we present numerical results generated by simulating the data phase of the communication protocol explained in \secref{sys}. We plot the probabilities of outage, transmission error and time overflow as a function of payload size for different training sequence lengths, transmit power and backoff values. Additionally, we compare the outage probability of the proposed variable-rate method with that of benchmark schemes.
 
\subsection{Simulation setup}
\paragraph*{Latency and reliability parameters} The most stringent motion control applications require a transmission time in the 0.25-1 ms range, while more tolerable processes might only require response times of the order of 1 ms \cite{neumann2007}. A typical factory automation system with up to 30 devices has a 1-2 ms transmission time, and a typical process control system with up to 200 field devices has a 10-15 ms transmission time \cite{schickhuber1997}. For our simulation, we choose a moderate aggregate transmission time of 1 ms. As for reliability, we choose a target outage probability of 10$^{-5}$.

\paragraph*{Miscellaneous parameters} We choose a conservative payload size of 50 Bytes (400 bits) per actuator, a total of 50 actuators, and an available bandwidth of 20 MHz. Putting together the number of actuators, their payload sizes, the available bandwidth and the transmission time, this translates into a \gls{dl} throughput of 1 bps/Hz. We suppose that the actuators and \glspl{ap} are scattered uniformly and independently over a floor area of $100\times 100$ m$^2$, an area similar to common areas of factories that have been surveyed in measurement campaigns \cite{hampicke1999,nist1951}.

\paragraph*{Path loss and blockage model} We assume a power law path loss model that gives power attenuation as a function of the distance from the transmitter, and a function of the link type: \gls{los} or \gls{nlos}. Let $\clos$ and $\cnlos$ be the reference path loss coefficients \gls{los} and \gls{nlos} links, and $\alos$ and $\anlos$ the path loss exponents. For a distance $r$, the path loss corresponding to \gls{los} and \gls{nlos} is
\begin{align}\label{eq:power_law}
\ell_{\text{L}}(r)=\clos r^{-\alos}, \; \ell_{\text{N}}(r)=\cnlos r^{-\anlos},
\end{align}
These four parameters have their values given in \cite{rappaport91}, also summarized in \tabref{t:params} along with the rest of the simulation parameter values. To determine the link type, we assume a blockage probability parametrized by a cutoff distance $d_0$ beyond which the probability of a link being \gls{los} becomes $p_0$:
\begin{align}
\plos(r) = p_0 + \indic{r\leq d_0}\,\frac{1-p_0}{d_0^2}\,(r-d_0)^2.
\end{align}

\paragraph*{Frequency band}We choose a carrier frequency of 3.5 GHz, also known as the \gls{cbrs}, where broadband networks in industrial applications are expected to deployed. Note that 20 MHz of bandwidth is consistent with the choice of the CBRS band, which is at least 50 MHz wide.

\begin{table}[t]
\centering
\caption{Values for simulation parameters.}
\label{t:params}
\begin{tabular}{ l  l}
\hline		
Parameter description & Value \\
\hline	
Transmission period, $T$ & 1 ms  \\
Number of actuators, $D$ & 50  \\
Number of \glspl{ap}, $A$  & 5 \\
Data per actuator, $B$ & 50 Bytes \\
Bandwidth, $W$ & 20 MHz \\
Floor area & 100$\times$100 m$^2$  \\
Backoff parameter & 0.8 \\
\gls{bs} transmit power & 23 dBm \\
Field actuator transmit power & 23 dBm \\
Noise power & -174 dBm/Hz \\
Carrier frequency & 3.5 GHz ($\lambda=$8.57 cm) \\
Path loss exponent ($\leq 10\lambda$) & 2 \\
LOS path loss exponent ($> 10\lambda$) & 3.26 \\
NLOS path loss exponent ($> 10\lambda$) & 3.93 \\
Blockage model: probability parameter $p_0$ & 0.25 \\
Blockage model: cutoff parameter $d_0$ & 15 m \\
\hline  
\end{tabular}
\end{table}

\subsection{Benchmark methods}\label{ss:bench}
\newcommand{\nomr}{\textsf{R}}		

We compare the outage probability for the variable-rate method to that for 4 benchmark methods described below.

\paragraph*{Cellular}\label{Sec:cellmodel}
Every \gls{ap} acts as a \gls{bs} in a cellular network with universal frequency reuse. Assuming equal cell sizes, every \gls{bs} is loaded with $D/B$ actuators, so every \gls{bs} broadcasts at a fixed rate $DB/AT$ bps.

\paragraph*{Fixed-rate} \label{Sec:wired_mbs}
According to the fixed-rate method, the \glspl{ap} use a single predetermined rate $DB/T$ bps to transmit to all of the actuators, one after another. An actuator fails to receive its data if its channel cannot support the predetermined rate (see our discussion of actuator failure in \secref{out}).

\paragraph*{Two-hop with cooperative relaying} 
This method takes place in two rounds. In the first round, one \gls{ap} uses a single predetermined rate $2DB/T$ bps to broadcast a packet with the data of all actuators; actuators who are able to decode the packet are deemed \textit{successful}. Should any actuators fail to receive the packet (suppose there are $D-k$ of those), the \gls{ap} cooperates with successful actuators to broadcast the original packet a second time at a reduced rate $2(D-k)B/T$ bps.

\subsection{Simulation Results}
 
\paragraph*{Training has opposite effects on the probabilities of transmission error and time overflow} With high quality channel estimation using long training sequences and an appropriate rate backoff, the probability that the chosen rate exceeds channel capacity decreases. This explains why the probability of transmission error decreases monotonically with the number of pilots as can be seen in \figref{fig:te_to}. Longer training sequences, however, increases channel sounding overhead, leaving less time for data transmission and increasing the chances of time overflow. This explains why the probability of time overflow monotonically increases with the number of pilots as can be seen in \figref{fig:te_to}. Since the probability of transmission error and the probability of time overflow have an opposite sense of variation with respect to length of training sequence, an optimal point corresponding to the minimal outage probability is expected to exist and does in fact exist as observed in \figref{fig:outage_vs_l}.

\begin{figure}[t]
\centering
\subfloat[$\prob{\evso}$]{
\includegraphics[width=0.4\textwidth]{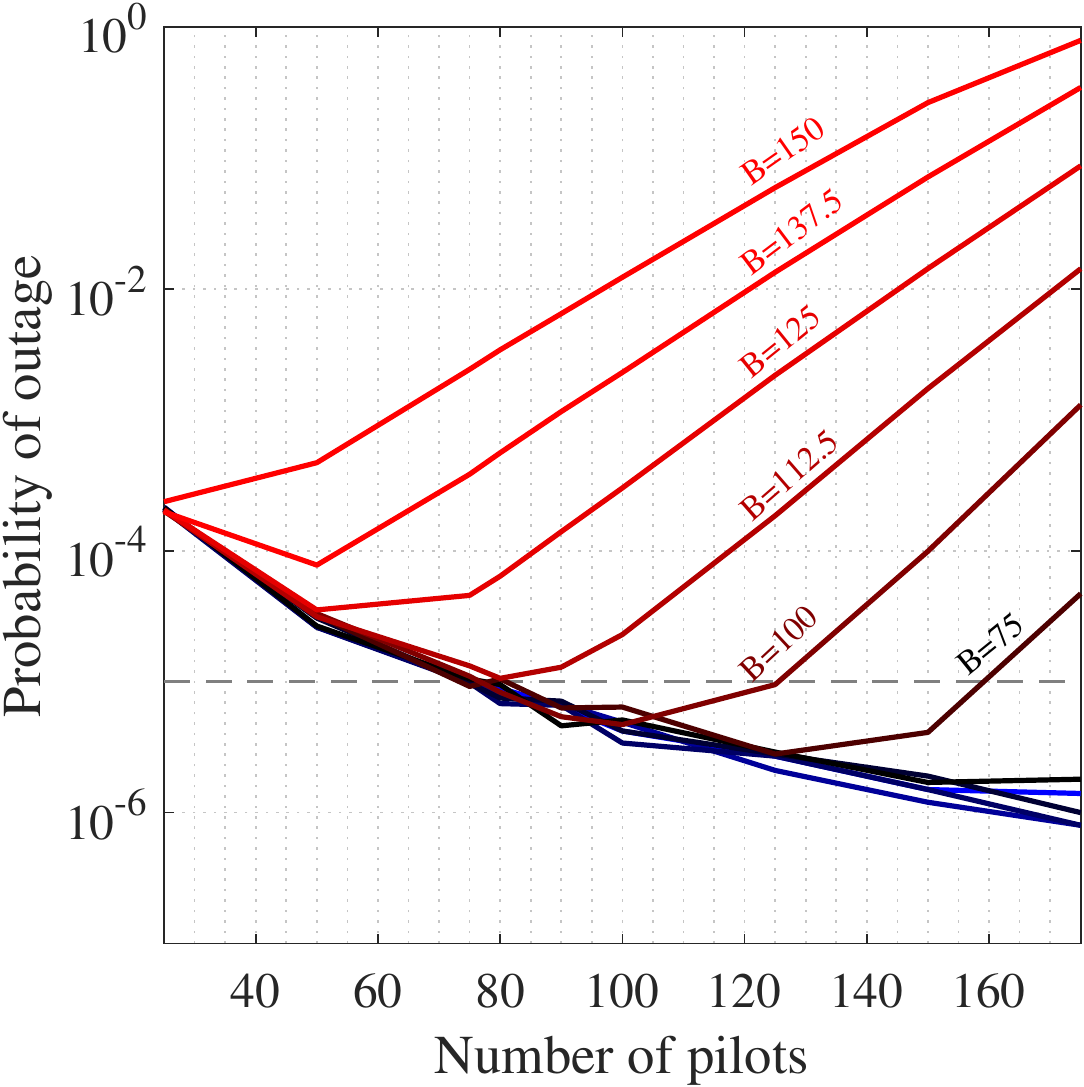}\label{fig:outage_vs_l}
}
\hfill
\subfloat[$\prob{\evte}$ and $\prob{\evto}$]{
\includegraphics[width=0.4\textwidth]{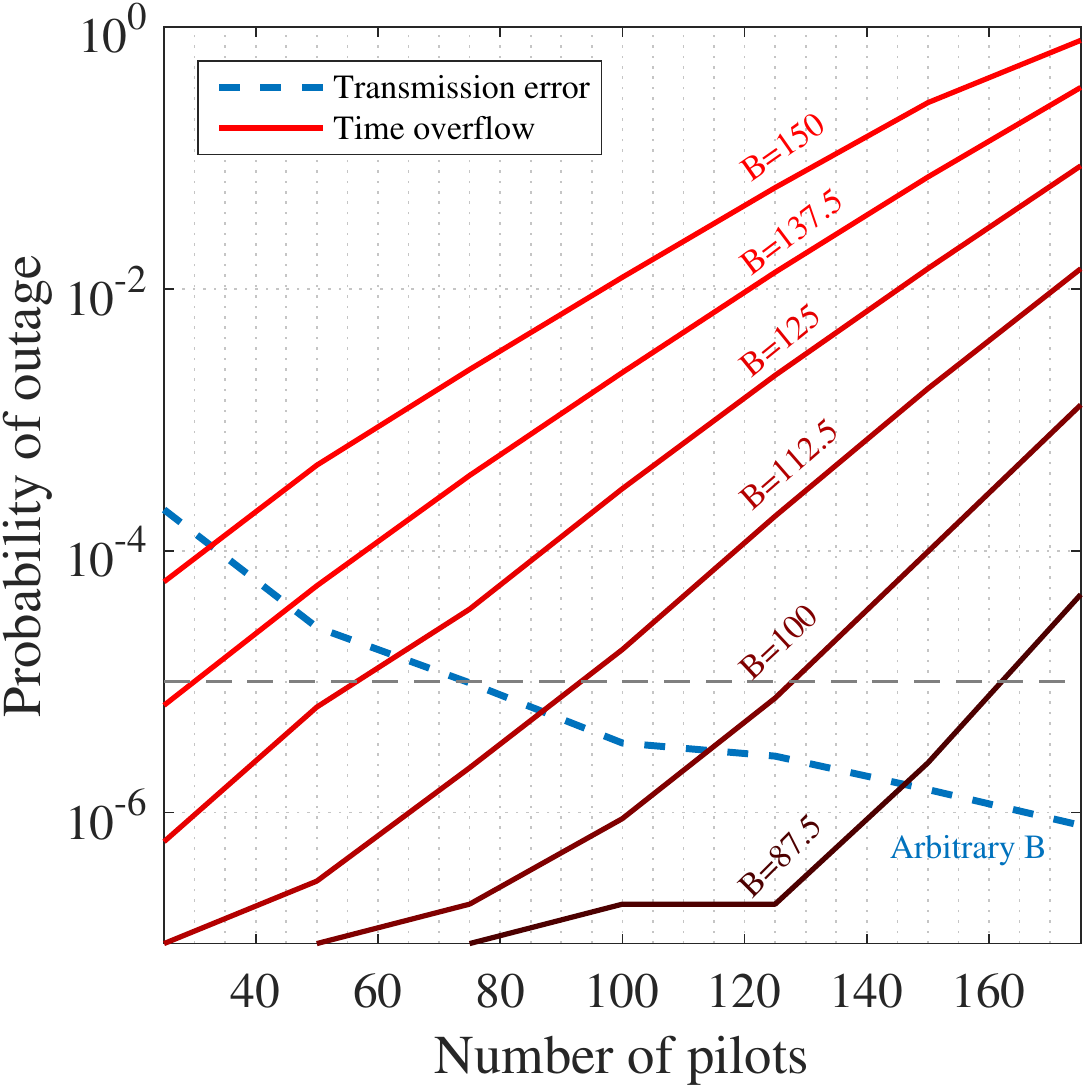}\label{fig:te_to}
}
\caption{\figref{fig:outage_vs_l} shows $\prob{\evso}$ as a function of the length of the training sequence per actuator for a range of payload sizes. This probability decreases till the number of pilots reach a \textit{training crossover point} after which it increases sharply. \figref{fig:te_to} shows $\prob{\evte}$ and $\prob{\evto}$. $\prob{\evte}$ decreases as the length of the training sequence increases. The probability curve for a payload of $50$ B only is included because curves for other payload sizes are almost identical. On the contrary, $\prob{\evto}$ increases as training becomes longer, suggesting that there is little time left for data transmission.}
\vspace{-20pt}
\end{figure}

\paragraph*{Outage is dominated by different events for different ranges of payload size} We can draw two further conclusions by jointly looking at \figref{fig:outage_vs_l}  and \figref{fig:te_to}, showing the probabilities of outage, transmission errors, and time overflow as a function of the length of the training sequence per actuator. Comparing \figref{fig:te_to} with \figref{fig:outage_vs_l}, we observe that outage is dominated by transmission error events for small payloads. Comparing \figref{fig:te_to} with \figref{fig:outage_vs_l}, we see that outage is dominated by time overflow for large payloads.

\begin{figure}[h]
\begin{minipage}[t]{0.4\textwidth}
\includegraphics[width=\linewidth,keepaspectratio=true]{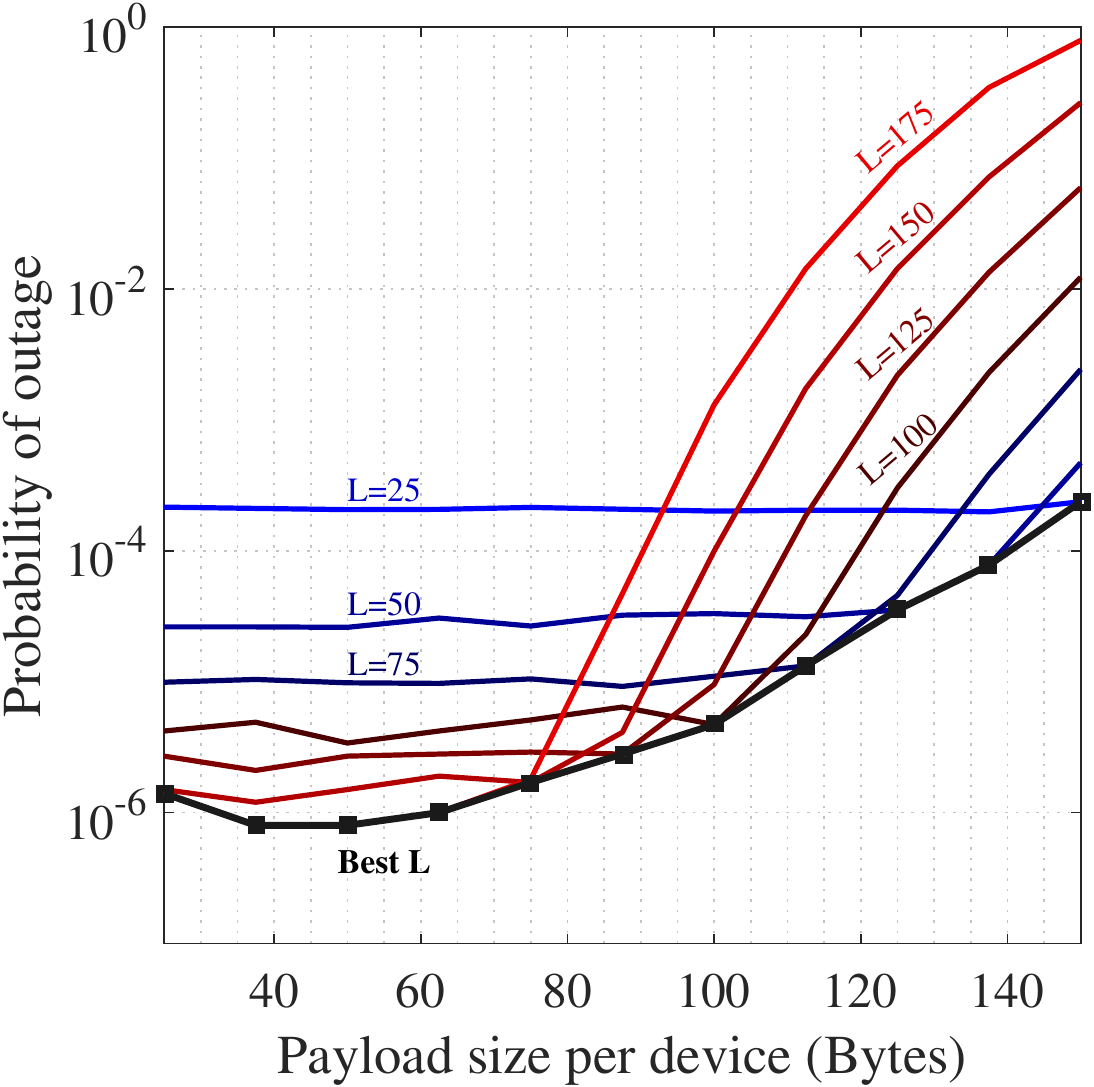}
\caption{The outage probability as a function of payload size per actuator for a range of training sequence lengths. This probability remains flat until the payload size reaches a \textit{payload crossover point}, beyond which the outage probability increases sharply.}
\label{fig:outage_vs_b}
\end{minipage}
\hspace*{\fill}
\begin{minipage}[t]{0.4\textwidth}
\includegraphics[width=\linewidth,keepaspectratio=true]{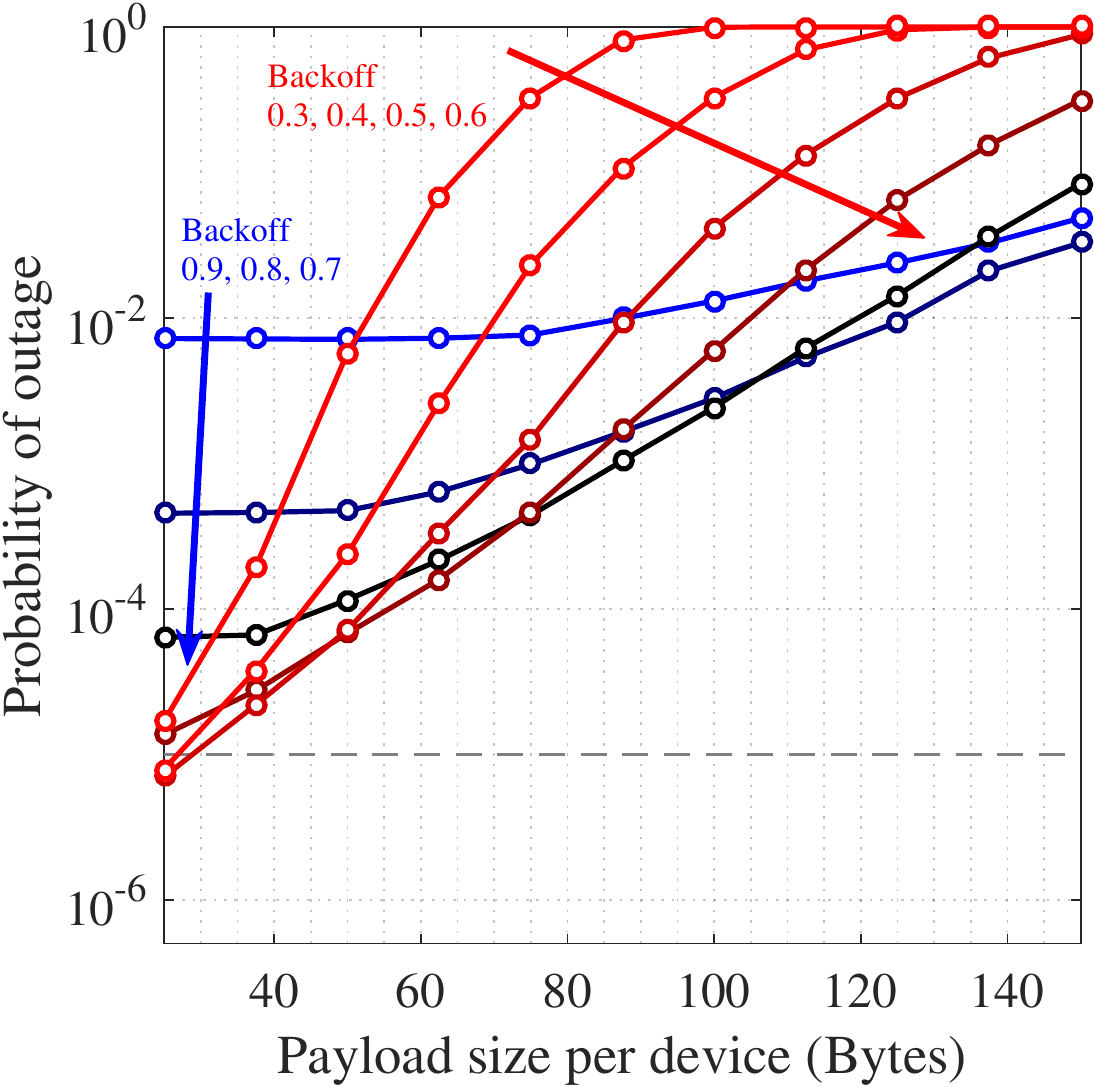}
\caption{The lowest outage probability as a function of payload size per actuator for different values of backoff.}
\label{fig:comp_backoff}	
\end{minipage}	
\end{figure}

\paragraph*{The main cause of outage sharply changes} Looking at \figref{fig:outage_vs_b} we observe that training decreases the outage probability for low payloads. However, we observe that training increases outage for high payloads. In relation to both \figref{fig:outage_vs_b} and \figref{fig:outage_vs_l}, we make the following observations. First, for a fixed training duration, the main cause of outage sharply transitions from transmission error to time overflow at a payload crossover point. Similarly, for a fixed payload size, the cause of outage sharply transitions from transmission error to time overflow at a training crossover point.

\paragraph*{There is an optimal backoff value for every payload size} \figref{fig:comp_backoff} shows the lowest outage probability as a function of payload size per actuator for the following values of backoff: 0.3, 0.4, 0.5, 0.6, 0.7, 0.8, and 0.9. For this plot, we considered 3 \glspl{ap} instead of the default 5 \glspl{ap}. The reason is that the lowest outage probability for some backoff-payload size pairs was exactly zero, i.e. there was not a single realization that produced an outage. In our simulations, we observe that for every payload size, there exists an optimal backoff value that minimizes the lowest outage probability. Put differently, for a particular payload size, the backoff parameter and the number of pilots can be chosen to reduce the outage probability.

\begin{figure}[h]
\centering
\includegraphics[width=0.4\textwidth]{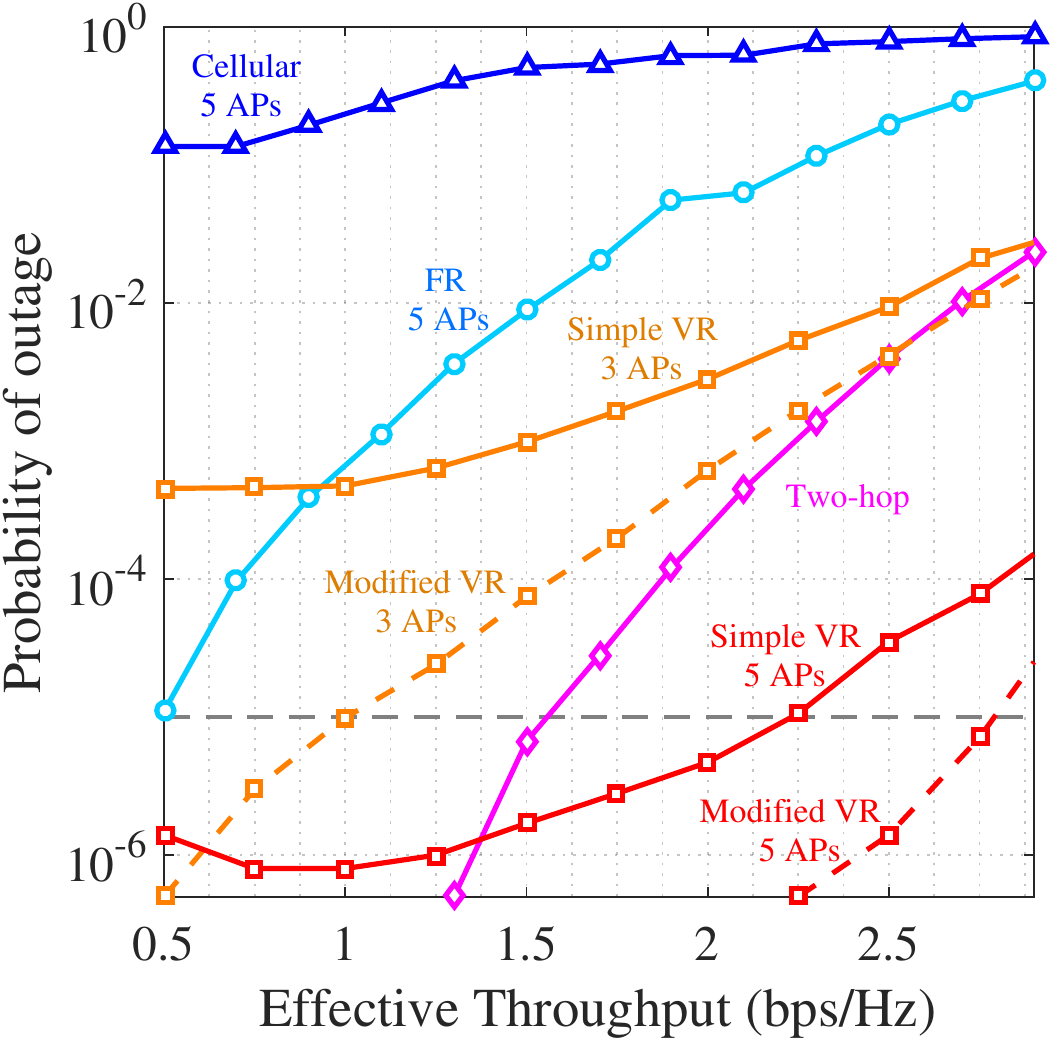}
\caption{Comparison of outage probabilities: simple variable-rate (VR), modified VR, fixed-rate (FR), cellular, and two-hop. In agreement with \secref{ss:reduce}, modified variable-rate has a lower outage probability than simple variable-rate.}
\vspace{-20pt}
\label{fig:comp1}
\end{figure}

\begin{figure}[h]
\centering
\subfloat[Comparison against SNR]{
\includegraphics[width=0.4\textwidth]{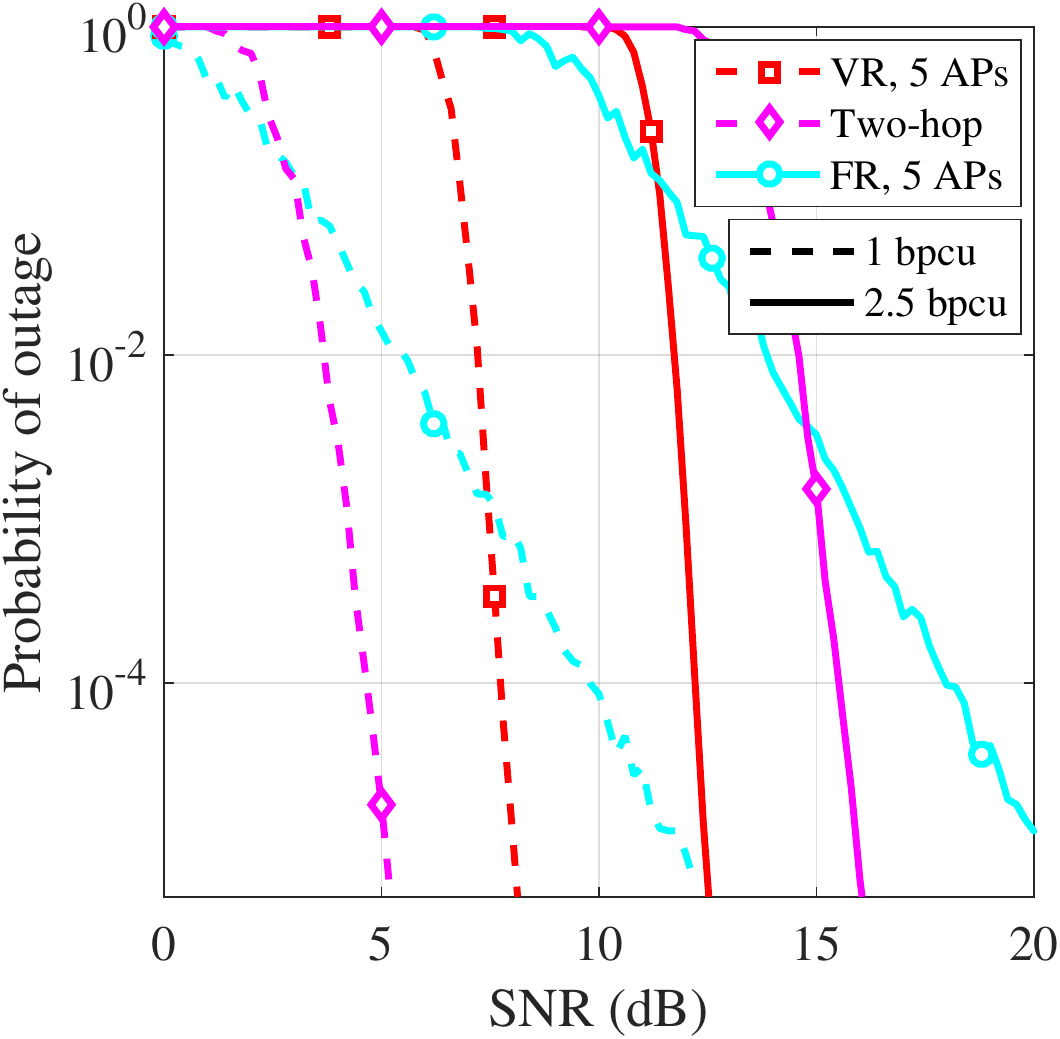}\label{fig:comp3}
}
\hfill
\subfloat[Finite-SNR diversity order]{
\includegraphics[width=0.4\textwidth]{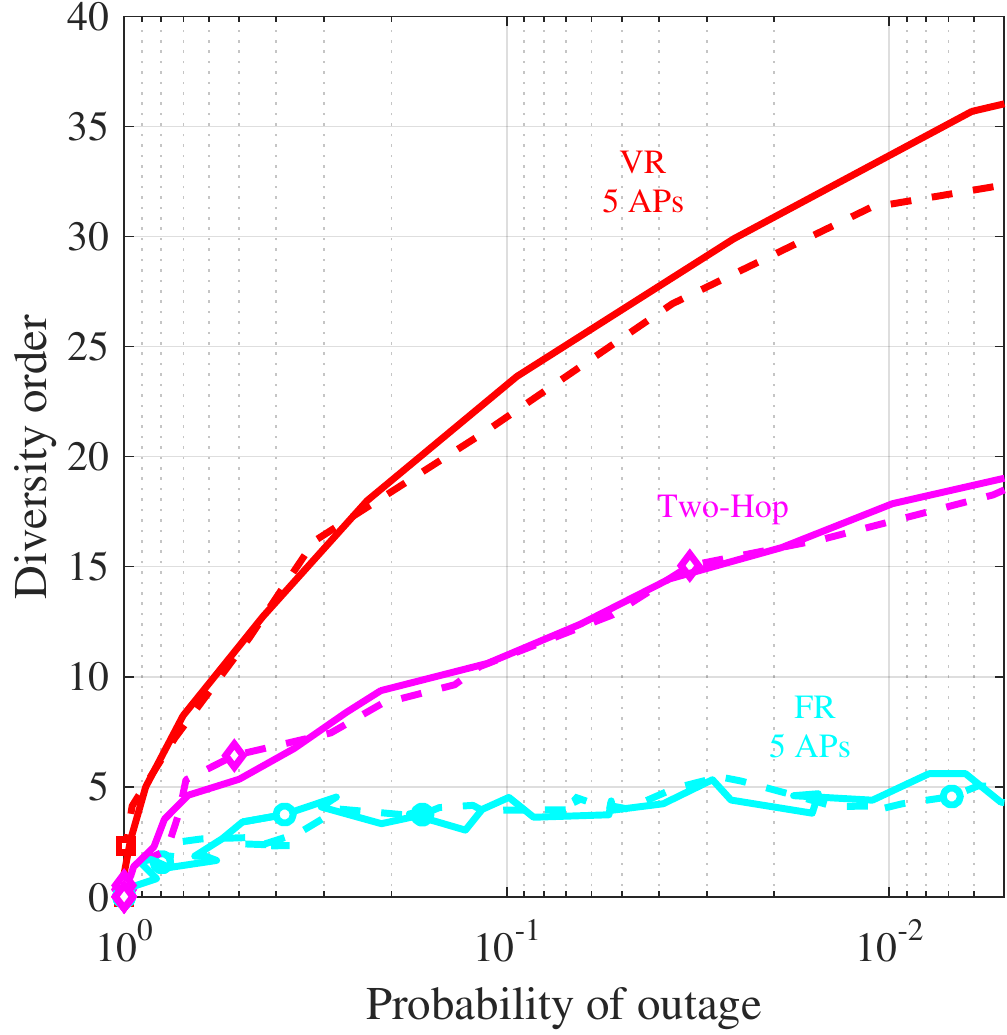}\label{fig:comp4}
}
\caption{\figref{fig:comp3} compares the outage probability of the variable-rate method, with that of the two-hop and fixed-rate schemes. The comparison is shown for two values of effective throughput, namely 1 and 2.5 bpcu. \figref{fig:comp4} compares the finite-SNR diversity order, derived based on the curves in \figref{fig:comp3}.}
\vspace{-20pt}
\end{figure}

\paragraph*{Modified variable-rate considerably outperforms simple variable-rate}
In agreement with \secref{ss:reduce}, modified variable-rate has lower outage probability than simple variable-rate as seen in \figref{fig:comp1}. Additionally, the spread of these probabilities widens as the payloads get smaller and shrinks as the payloads get bigger. For 3 \glspl{ap}, the spread is about 3 orders of magnitude for 20 B payloads. For 5 \glspl{ap}, the spread appears to be significantly wider.

\paragraph*{The adaptive-rate method outperforms benchmark methods}
\figref{fig:comp1} shows the outage probability as a function of effective throughput for the VR methods in comparison to benchmark methods introduced in \secref{ss:bench}. The effective throughput is defined as the total number of bytes transmitted by the \glspl{ap} divided by the total transmission duration.  
For VR, we plot the outage probability for 3 \glspl{ap} and 5 \glspl{ap}. For the cellular and FR we plot the outage probability for $A=5$. We observe that modified VR for $A=5$ outperforms all benchmarks and comes within the 10$^{-5}$ target outage probability.

\paragraph*{The adaptive-rate method has a larger diversity gain}
To focus on diversity gain of the transmission schemes, we adopt a simulation setup where all actuators exhibit a single nominal average \gls{snr}. \figref{fig:comp3} shows the outage probability against nominal \gls{snr}. We compare the outage probabilities between modified VR and benchmarks for an effective throughput of 1 and 2.5 bpcu, corresponding to $50$ B and $125$ B payloads. The corresponding empirical diversity orders, defined as the local negative slope of outage curves in log-log scale, are plotted in \figref{fig:comp4} against the outage probability. Modified VR requires considerably smaller transmit power to reach the target outage probability.  The comparison shows that regardless of packet size, the VR achieves a larger diversity gain compared to the two-hop method. This is an interesting observation given that in the two-hop method, all successful actuators help increase the cooperative diversity gain in the relaying phase whereas in VR has the cooperative diversity gain from only 5 \glspl{ap} and it is mainly reliant on multi-user scheduling to achieve a large diversity order. Comparing the diversity order of VR to that of FR is also insightful. Relying on cooperation of 5 \glspl{ap}, the diversity order of FR method quickly saturates at 5 while VR reaches a significantly higher diversity order.

\section{Conclusion}\label{conc}
We have modeled a multi-user, variable-rate communications system for industrial applications that require data to be received within a hard deadline. Motivated by the need to exploit spatial diversity to achieve ultra reliability, we proposed a pilot-assisted variable-rate method that exploits multi-user diversity. We have studied the probability of outage of the industrial system as caused by two events: transmission error due to the transmission at a rate that cannot be supported by the channel, and time overflow due to the inflexible downlink time budget. We have determined closed-form expressions for the probability of transmission error, a closed-form expression for the probability of time overflow when there are two field devices, and a number of upper and lower bounds for the probability of time overflow for an arbitrary number of field devices. Our simulation results have revealed several interesting phenomena. Longer training decreases the probability of transmission error but increases that of time overflow. Longer training decreases the probability of outage for small payloads but increases this probability for large payloads because payload size is proportional to  airtime. If the latter is too long, the downlink budget might be insufficient and an outage could occur due to time overflow. We have observed that for a fixed training sequence length, the main source of outage sharply changes at a payload crossover point. We have also observed that for a fixed payload, the main source of outage sharply changes at a training crossover point. These two observations suggest that there is an optimal training length for a set payload size, and an optimal payload size for a set training sequence length. In particular, our results show that in spite of the stringent delay constraint, it pays to learn the channel through pilot transmissions and adapt the transmission rate according to the channel quality. Adapting the transmission rate and hence the transmission duration for each device makes the outage probability a function of the channels of all device-\gls{ap} pair and hence multi-user diversity is exploited to achieve high reliability.

\appendices
\section{}\label{a:a}
\paragraph{Proof of \propref{p:err_1r_1t}}
Now, we remove the conditioning on the knowledge of $\hhat$ to determine $\prob{|\hhat|>|h|}$. Let $f$ be the \gls{pdf} of $|\hhat|^2$.
\begin{align}
\prob{|\hhat|>|h|} 
&= \expect{\prob{|\hhat|>|h| \given \hhat} } \\
&= \frac{1}{2} - \frac{1}{2}\int\limits_0^\pinf e^{-\frac{2}{\sige^2}}\,I_0\left(\frac{2}{\sige^2} u\right)\,
\frac{1}{1-\sige^2} e^{-\frac{1}{1-\sige^2} u} \diff u \\
&= \frac{1}{2} - \frac{1}{2(1-\sige^2)}  \int\limits_0^\pinf I_0\left(\frac{2}{\sige^2} u\right)\,
e^{-\frac{2-\sige^2}{\sige^2(1-\sige^2)}u} \diff u \\
&= \frac{1}{2} - \frac{1}{2(1-\sige^2)} \cdot 
\frac{1}{\sqrt{\left[\frac{2-\sige^2}{\sige^2(1-\sige^2)}\right]^2 - \left[\frac{2}{\sige^2}\right]^2}},
%
\end{align}
which leads to the result in \corref{p:err_1r_1t}.

\paragraph{Proof of \propref{p:err_1r_Dt}}
\newcommand{\rebalw}{{\frac{2}{\sige^2}\norm{\bfhhat}^2}}
Now redefine $R\triangleq \logof{1+\rho\norm{\bfhhat}^2}$. We first determine $\prob{R>C\given \bfhhat}$, the probability that $R>C$. This is equivalent to $\prob{\norm{\bfhhat}^2>\norm{\bfh}^2 \given \bfhhat}$, which is given as
\begin{align}
\prob{\norm{\bfhhat}>\norm{\bfh} \given \bfhhat} 
&= 1-\Qmof{\frac{\sqrt{2}}{\sige}\norm{\bfhhat},\frac{\sqrt{2}}{\sige}\norm{\bfhhat|}}{A}
\\
&= \frac{1}{2}\left[1-e^{-\rebalw}I_0\left(\rebalw\right)\right]
- e^{-\rebalw} \, \sum_{k=1}^{A-1} I_k\left(\rebalw\right),
\end{align}
where the last equality is given in \cite{marcum}.

\renewcommand{\rebalw}{{\frac{2}{\sige^2}u}}
\noindent Note that $\norm{\bfhhat}^2\sim\text{Gamma}\left(A,\frac{1}{1-\sige^2}\right)$. Redefine $f$ to be the \gls{pdf} of $\norm{\bfhhat}^2$. Hence,
\begin{align}
\prob{\norm{\bfhhat}>\norm{\bfh}}
&= \frac{1}{2} - \frac{1}{2}\int\limits_0^\pinf \frac{1}{\Gamma(A)\left(1-\sige^2\right)^A}\, u^{A-1}\, I_0\left(\rebalw\right)\,e^{-\ctep u} \diff u \\
&+ \;\;\;\;\; \sum_{k=1}^{A-1} \int\limits_0^\pinf \frac{1}{\Gamma(A)\left(1-\sige^2\right)^A}\, u^{A-1}\, I_k\left(\rebalw\right)\,e^{-\ctep u} \diff u.
\end{align}

\noindent These two integrals can be seen as Laplace transforms of transcendental functions, and \textit{Tables of Integrals and Transforms, Vol. I} \cite{bateman1}, gives them the following closed form expression:
\begin{align}
\lapvarof{u^{\mu}\,I_\nu(au)}{p} 
&= \Gamma(\mu+\nu+1) \, \left(p^2-a^2\right)^{-\mu-1} P_\mu^{-\nu}\left(p/\sqrt{p^2-a^2}\right),
\end{align}
where $P_\mu^{-\nu}$ is the associated Legendre function for all $\mu,\nu$ that satisfy $\text{Re}\{\mu+\nu\}>-1$, and all $a,p$ that satisfy $p>a$. After some simplification, we finally obtain the result in \propref{p:err_1r_Dt}.

\section{}\label{a:b}
\paragraph{Proof of \propref{p:to_2r}}
We have defined $R=\log(1+|\hhat|^2\rho)$. Let $Y\triangleq 1/R$, and let $F_Y$ and $f_Y$ be the \gls{cdf} and \gls{pdf} of $Y$. For $y\geq 0$,
\begin{align}
F_Y(y) 
&= 1-\prob{R\leq \frac{1}{y}} \\
&= e^{ -\frac{1}{\rho\varhh}\left(2^{1/y}-1\right) }.
\end{align}
\noindent Let $\rho'\triangleq \rho\varhh$ be the measurement \gls{snr} and $G\triangleq 1/\rho'$. By taking the first derivative of $F_Y$, it follows that the \gls{pdf} is
\begin{align}
f_Y(y) = \left(G e^{G }\ln 2\right) \, \frac{2^{\frac{1}{y}}e^{-G 2^{1/y}}}{y^2}, \; y\geq 0.
\end{align} 

\renewcommand{\rebalw}{\left(G e^{G}\ln 2\right)}
\newcommand{\rebaltw}{\left(G e^{2G}\ln 2\right)}
\noindent Now let $Y_1, Y_2 \sim Y$. For $y\geq 0$
\begin{align}
\prob{Y_1+Y_2\leq y}
&= \int\limits_0^y f_Y(x) F_Y(y-x) \diff x \\
&= \rebalw \int\limits_0^y e^{ -G\left(2^{1/y-x}-1\right) } \;
2^{\frac{1}{x}}e^{-G 2^{1/x} } \frac{1}{x^2} \diff x \\
&= \rebaltw \int\limits_0^y e^{ -G\left(2^{1/x} + 2^{1/y-x}\right) } \;
2^\frac{1}{x} \; \frac{1}{x^2} \diff x
\intertext{with the change of variables $v=2^{1/x}$,}
&= \rebaltw \int\limits_{2^{1/y}}^\pinf \expof{-G\left(v+2^{\frac{1}{y-\ln 2/\ln v}}\right)} \, v\, \left(\frac{\ln v}{\ln 2}\right)^2 \frac{\ln 2}{v (\ln v)^2} \, \diff v \\
&= \frac{1}{\rho\varhh} \int\limits_{2^{1/y}}^\pinf
\expof{ -\frac{v+v^{1/y\log v - 1}-2}{\rho\varhh} } \, \diff v .
\end{align}

\noindent Noting that $v^{1/y\log v - 1}\geq v^{1/y\log v}=2^{1/y}$, we get the upper bound in \propref{p:to_2r}. 

\noindent As for the lower bound, we first note that the function $\varphi(v)=\expof{-Gv^{1/y\log v}}$ is strictly increasing over $(2^{\frac{1}{y}},\pinf)$ since it is a composition of two strictly decreasing functions (note that the power function is strictly increasing). For a small $\epsilon>0$, we write
\begin{align}
\prob{Y_1+Y_2\leq y} 
&= G\int\limits_{2^{1/y}}^\pinf 
\expof{ -G\left( v+v^{1/y\log v - 1}-2 \right) } \, \diff v \\
&\geq Ge^{2G}\int\limits_{2^{1/y+\epsilon}}^\pinf
\expof{ -G\left( v+v^{1/y\log v - 1} \right) } \, \diff v \\
&\geq Ge^{2G}\int\limits_{2^{1/y+\epsilon}}^\pinf
\expof{ -G\left( v+(2^{1/y+\epsilon})^{1/y\epsilon} \right) } \, \diff v,
\intertext{taking $\epsilon=1/y$, we get}
&= Ge^{2G-G2^{2/y}}\int\limits_{2^{1/y+\epsilon}}^\pinf
e^{-Gv} \, \diff v,
\end{align}
which leads to the lower bound given in \propref{p:to_2r}.

\paragraph{Proof of \propref{p:to_nr}}
Let $q_n\triangleq \prob{\sum_{k=1}^n Y_k \leq y}$. We assume that
\begin{align}
q_{n-1}(y) \leq \expof{-(n-1)\,G\left(2^{1/y}-1\right)}.
\end{align}
This is true for orders $1$ and $2$, i.e $n-1=1,2$. To prove for a general $n$, we prove the inductive step as follows:
\begin{align}
q_n(y)
&= \int\limits_0^y  q_{n-1}(x) f_Y(y-x)  \diff x \\
&\leq \rebalw \int\limits_0^y \expof{-(n-1)\,G\left(2^{1/x}-1\right)} \;
2^{\frac{1}{y-x}}e^{-G(2^{1/y-x})} \frac{1}{(y-x)^2} \diff x\\
\intertext{with the change of variables $v=2^{1/(y-x)}$,}
&= \left(G e^{nG}\right) \int\limits_{2^{1/y}}^\pinf \;
\expof{- (n-1)\,G  2^\frac{\log v}{y\log v -1} } \, e^{-G v} \diff v \\
&\leq  \left(G e^{nG}\right) \int\limits_{2^{1/y}}^\pinf \;
\expof{- (n-1)\,G \, 2^\frac{1}{y} } \, e^{-G v} \diff v \\
&\leq G \; \expof{- (n-1)\,G \, 2^\frac{1}{y} + n\,G} \;
\int\limits_{2^{1/y}}^\pinf e^{-G v} \diff v,
\end{align}
which leads to the upper bound. As for the lower bound, we assume that 
\begin{align}
q_{n-1}(y) \leq \expof{-(n-1)\,G\left(2^{(n-1)/y}-1\right)}.
\end{align}
Note that this is also true for orders $1$ and $2$. We show the inductive step as follows:
\begin{align}
q_n(y)
&\geq \rebalw \int\limits_0^y \expof{-(n-1)\,G\left(2^{(n-1)/x}-1\right)} \;
2^{\frac{1}{y-x}}e^{-G(2^{1/y-x})} \frac{1}{(y-x)^2} \diff x\\
\intertext{with the change of variables $v=2^{1/(y-x)}$,}
&= \left(G e^{nG}\right) \int\limits_{2^{1/y}}^\pinf \;
\expof{- (n-1)\,G  v^\frac{n-1}{y\log v -1} } \, e^{-G v} \diff v \\
&\geq \left(G e^{nG}\right) \int\limits_{2^{1/y+\epsilon}}^\pinf \;
\expof{- (n-1)\,G  v^\frac{n-1}{y\log v -1} } \, e^{-G v} \diff v \\
&\geq \left(G e^{nG}\right) \int\limits_{2^{1/y+\epsilon}}^\pinf \;
\expof{- (n-1)\,G  (2^{1/y+\epsilon})^\frac{n-1}{y\epsilon} } \, e^{-G v} \diff v, \\
\intertext{taking $\epsilon=(n-1)/y$,}
&= \expof{- (n-1)\,G  2^{n/y} +nG}
\int\limits_{2^{n/y}}^\pinf  e^{-G v} \diff v,
\end{align}
which leads to the lower bound.

\section{}\label{a:c}
\paragraph{Proof of \propref{p:to_2r_bnd}}
First, we make the following claim:
\begin{equation}\label{eq:bold_claim}
e^{-G v^\frac{1}{y\log v-1}} \leq e^{-G 2^{1/y}} \, \left( 1 - e^{-A(v-2^{1/y})} \right),\;\text{for all $v\geq 2^{1/y}$},
\end{equation}
for some $A \equiv A(y,\,G) \geq 0$ to be determined. We proceed to derive the bound on $\prob{Y_1+Y_2\leq y}$ starting with \eqref{eq:to_2r}.
\begin{align}
\prob{Y_1+Y_2\leq y} 
&= G \int\limits_{2^{1/y}}^\pinf
\expof{ -G\left(v+v^{1/y\log v - 1}-2\right) } \, \diff v \\
&\leq G e^{2G}e^{-G 2^{1/y}} \int\limits_{2^{1/y}}^\pinf 
e^{-G v}\,\left(1 - e^{-A(v-2^{1/y})}\right) \,\diff v \\
&= G e^{2G}e^{-G 2^{1/y}} \int\limits_{2^{1/y}}^\pinf e^{-G v} \,\diff v
\;\;+\;\; G e^{2G}e^{(A-G)2^{1/y}} \int\limits_{2^{1/y}}^\pinf e^{-(A+G)v} \,\diff v \\
&= e^{2G}e^{-2G 2^{1/y}} - \left(\frac{G}{A+G}\right) e^{2G} e^{(A-G)2^{1/y}} e^{-(A+G)2^{1/y}} \\
&= \left(1-\frac{G}{A+G}\right) e^{2G} e^{-2G 2^{1/y}}.
\end{align}
Finally,
\begin{equation}
\prob{Y_1+Y_2\leq y}  \leq \left(\frac{A}{A+G}\right) e^{-2G(2^{1/y}-1)}.
\end{equation}

Now that we have tightened our bound on the probability, we prove the claim that we have made. We first rewrite the inequality in \eqref{eq:bold_claim} as
\begin{align}
\expof{-2^\frac{\ln v}{y\ln v - \ln 2}} 
&\leq e^{-G 2^{1/y}} - \expof{-v\left(A-\frac{(A-G)2^{1/y}}{v}\right)}   \\
\expof{-v\left(A-\frac{(A-G)2^{1/y}}{v}\right)}  
&\leq  e^{-G 2^{1/y}} - \expof{-2^\frac{\ln v}{y\ln v - \ln 2}} \\
\expof{-v\,\Psi(v)} &\leq \expof{-v\,\Phi(v)}.
\end{align}
Now, we determine $\Phi(v)$.
\begin{align}
\expof{-v\,\Phi(v)}
&= \expof{\lnof{e^{-G 2^{1/y}} - \expof{-2^\frac{\ln v}{y\ln v - \ln 2}}}} \\
&= \expof{-v\cdot \frac{-1}{v}\lnof{e^{-G 2^{1/y}} - \expof{-2^\frac{\ln v}{y\ln v - \ln 2}}}}.
\end{align}
Since the exponential function is increasing, a necessary condition for \eqref{eq:bold_claim} is $\Psi(v)\geq \Phi(v)$, for all $v\geq 2^{1/y}$. After a few steps of calculation, we obtain the following condition:
\begin{align}
A \geq \left(\frac{1}{2^{1/y}-v}\right)\,\lnof{1- \expof{G 2^\frac{1}{y}-G 2^\frac{\ln v}{y\ln v - \ln 2}}}.
\end{align}

\noindent We still need to prove that such an $A$ does exist. We write the \gls{rhs} as a product of two functions $f$ and $g$, where $f$ is negative and monotonically increasing towards $0$, and $g$ is also negative but monotonically decreasing towards $0$. It is clear that $\lim_{v\rightarrow 2^{1/y}}=0$. To show that $\lim_{v\rightarrow \,+\infty}=0$, we replace $g$ with a function $h$ that goes faster to $\ninf$ as its argument goes to $0$. One candidate is $h(x)=-1/x$ because $\lim_{x\rightarrow 0} \ln(x)/(-1/x) = 0$. Hence,
\begin{align}
\lim_{v\rightarrow\,\pinf} f(v)g(v)
&=\lim_{v\rightarrow\pinf} \left(\frac{1}{2^{1/y}-v}\right)\,\lnof{1- \expof{G 2^\frac{1}{y}-G 2^\frac{\ln v}{y\ln v - \ln 2}}} \\
&\leq 
\lim_{v\rightarrow\,\pinf} \frac{1}{G 2^{1/y}-v}\cdot\frac{-1}{{1- \expof{G 2^\frac{1}{y}-G 2^\frac{\ln v}{y\ln v - \ln 2}}}} \\
&=
\lim_{v\rightarrow\,\pinf} \frac{1}{v}\cdot\frac{e^{-G 2^{1/y}}}{e^{-G 2^{1/y}} - \expof{-G 2^\frac{\ln v}{y\ln v - \ln 2}}} \\
&=
\lim_{v\rightarrow\,\pinf} \frac{1}{v}\cdot\frac{1}{e^{-G 2^{1/y}} - \expof{-G 2^\frac{\ln v}{y\ln v - \ln 2}}} \\
&=
\lim_{x\rightarrow\,\pinf} \frac{e^{-x}}{e^{-G 2^{1/y}} - \expof{-G 2^\frac{x}{yx - \ln 2}}}.
\end{align}
By using l'Hopital's rule, we get
\begin{align}
\lim_{v\rightarrow\,\pinf} f(v)g(v)
&=
\lim_{x\rightarrow\,\pinf} \frac{e^{-x}\;(yx-\ln  2)^2}{ G(\ln 2)^2 \, 2^{\frac{x}{yx-\ln 2}}\, e^{-G 2^{\frac{x}{yx-\ln 2}}}} \phantom{\;\;\;\;\;\;\;}\\
&=\frac{\lim_{x\rightarrow\,\pinf} e^{-x}\;(yx-\ln  2)^2}{G(\ln 2)^2 \, 2^\frac{1}{y} \, e^{-G 2^{1/y}}} \; = 0.
\end{align}

\noindent Since $fg$ does not admit any singularities on $(2^{1/y},\pinf)$, $fg$ is differentiable on that interval. Since, additionally, $\lim_{v\rightarrow 2^{1/y}} f(v)g(v)$  $=\lim_{v\rightarrow\,\pinf} f(v)g(v)=0$, Rolle's theorem says that $fg$ must have a critical point on $(2^{1/y},\pinf)$. Finally, since $fg\geq 0$, we conclude that that critical point is in fact a maximum; we choose $A$ to be that.

\section{}\label{a:d}
\paragraph{Proof of \propref{p:sandwich}}
We first show that the lower bound holds. If for all $d$, $1/R_d > \mytd/DB$, then clearly $\sum_{d=1}^D 1/R_d > \sum_{d=1}^D \mytd/DB = \mytd/B$. Therefore
\begin{align}
\begin{split}
\prob{\sum_{d=1}^D \frac{1}{R_d} > \frac{\mytd}{B}}
&\geq
\prob{\bigcap_{d=1}^D \left\{\frac{1}{R_d} > \frac{\mytd}{DB}\right\}} \\
&= \prod_{d=1}^D \prob{R_d < \frac{DB}{\mytd}},
\end{split}
\end{align}
where the last equality follows because $\set{R_d}$ are independent random variables. In a similar way, we show that the upper bound holds. If $\sum_{d=1}^D 1/R_d > \mytd/B$, then $\exists d$ s.t. $R_d < DB/\mytd$. To see that this is true, assume otherwise. If $R_d \geq DB/\mytd$ for all $d$, then $\sum_{d=1}^D 1/R_d \leq \mytd/B$, a contradiction. Therefore
\begin{align}
\begin{split}
\prob{\sum_{d=1}^D \frac{1}{R_d} > \frac{\mytd}{B}}
&\leq
\prob{\bigcup_{d=1}^D \left\{R_d < \frac{DB}{\mytd} \right\}} \\
&=
1-\prod_{d=1}^D \prob{\bigcap_{d=1}^D \left\{R_d > \frac{DB}{\mytd} \right\}}.
\end{split}
\end{align}
\bibliographystyle{IEEEtran}
\bibliography{IEEEabrv,ref}
\end{document}